\documentclass{sig-alternate-2013} 

\usepackage{amsmath}
\usepackage{amsfonts}
\usepackage{amssymb}
\usepackage{balance}
\usepackage{booktabs}
\usepackage{enumerate}
\usepackage{cite}
\usepackage[usenames,dvipsnames]{xcolor}
\usepackage{graphicx}
\usepackage{mdwlist}    % from christos / mary: makes lists tighter
\usepackage{mathtools}
\usepackage{multirow}
\usepackage{subfigure}
\usepackage{times}
\usepackage{url}
\usepackage{xspace} 
\usepackage{enumitem} 

\newtheorem{theorem}{Theorem}[section]

\def\pprw{8.5in}
\def\pprh{11in}

\newcommand{\hide}[1]{}

\newcommand{\bit}{\begin{itemize*}}
\newcommand{\eit}{\end{itemize*}}
\newcommand{\ben}{\begin{enumerate*}}
\newcommand{\een}{\end{enumerate*}}

\newcommand{\lowrank}{{latent enviroments \xspace}}
\newcommand{\model}{{TribeFlow}\xspace}
\newcommand{\Hentity}{{Users}\xspace}
\newcommand{\SDentity}{{Items}\xspace}
\newcommand{\hentity}{{users}\xspace}
\newcommand{\sdentity}{{items}\xspace}

\DeclareMathAlphabet\mathbfcal{OMS}{cmsy}{b}{n}

\newcommand{\cM}{\mathcal{M}}

\newcommand{\cD}{\mathcal{D}}
\newcommand{\cU}{\mathcal{U}}

\newcommand\Mark[1]{\textsuperscript{#1}}

\permission{Copyright is held by the International World Wide Web Conference Committee (IW3C2). IW3C2 reserves the right to provide a hyperlink to the author's site if the Material is used in electronic media.}
\conferenceinfo{WWW 2016,}{April 11--15, 2016, Montr\'eal, Qu\'ebec, Canada.}
\copyrightetc{ACM \the\acmcopyr}
\crdata{978-1-4503-4143-1/16/04. \\
http://dx.doi.org/10.1145/2872427.2883059}
\begin{document}
\title{TribeFlow: Mining \& Predicting User Trajectories}
\numberofauthors{1}
\author{
Flavio Figueiredo\Mark{1,2},\, Bruno Ribeiro\Mark{4,5},\, Jussara Almeida\Mark{3},\, Christos Faloutsos\Mark{5}\vspace{3pt}\\
\affaddr{\Mark{1}UFCG - Brazil,}\quad%Universidade Federal de Campina Grande - Brazil,} \quad
\affaddr{\Mark{2}IBM Research - Brazil,}\quad%Universidade Federal de Campina Grande - Brazil,} \quad
\affaddr{\Mark{3}UFMG - Brazil,}\quad%Universidade Federal de Minas Gerais - Brazil,}\\
\affaddr{\Mark{4}Purdue University,} \quad
\affaddr{\Mark{5}Carnegie Mellon University}\vspace{3pt}\\
\email{\{flaviov,jussara\}@dcc.ufmg.br,} \email{ribeiro@cs.purdue.edu}, \email{christos@cs.cmu.edu}
}

\maketitle

\begin{abstract}
%!TEX root = paper.tex
Which song will Smith listen to next? 
Which restaurant will Alice go to tomorrow? 
Which product will John click next?
These applications have in common the prediction of user trajectories that are in a constant state of flux over a hidden network (e.g.\ website links, geographic location).
But what users are doing now may be unrelated to what they will be doing in an hour from now.
Mindful of these challenges we propose \model, a method designed to cope with the complex challenges of learning personalized predictive models of non-stationary, transient, and time-heterogeneous user trajectories. 
\model is a general method that can perform next product recommendation, next song recommendation, next location prediction, and general arbitrary-length user trajectory prediction without domain-specific knowledge.
\model is more accurate and up to ${\bf 413\times}$ {\bf faster} than top competitors.

%User attention is arguably one of the most scarce and vied for commodities of today's Internet economics.
 %Selling user attention is a multi-billion dollar business that   sustains some of the most popular online services.
%Understanding how users access and exchange information is a central task in unveiling Internet user behavior.
%However, mining information flows exchanged between a set of entities (users-to-user, users-to-products) poses many challenges to current tensor decomposition methods.
%Information flows often lack stationarity, ergodicity, and a characteristic time scale.
%Moreover, available data is often biased and the population of most interest might be underrepresented.
%To perform tensor decomposition accounting for all these factors we propose \model, a method that reveals interesting and meaningful maps of information flows between heterogeneous interacting entities.
%Specifically, we  observe that overall user attention seem to  be elastic, as many artists seem to cooperate for user attention,
%% specially the more established ones, 
%although we also find evidence of competition, notably newcomers that steal user attention away from others.
%We apply our technique to ...

\end{abstract}

%\category{H.3.5}{Information Storage and Retrieval}{Online Information Services}[Web-based services]
%\terms{Algorithms; Measurement}

\section{Introduction}
\label{sec:intro}
%!TEX root = paper.tex

Web users are in a constant state of flux in their interactions with products, places, and services.
User preferences and the environment that they navigate determine the sequence of items that users visit (links they click, songs they listen, businesses they visit).
In this work we refer to the sequence of items visited by a user as the user's trajectory.
Both the environment and user preferences affect such trajectories.
The underlying navigation environment may change or vary over time: a website updates its design, a suburban user spends a weekend in the city.
Similarly, user preferences may also vary or change over time: a user has different music preferences at work and at home, a user prefers ethnic food on weekdays but will hit all pizza places while in Chicago for the weekend.
%Even in shorter time scales our behavior in the morning may be different than what we do in the afternoon. 

The above facts result in user trajectories that over multiple time scales can be non-stationary (depend on wall clock times), transient (some visits are never repeated), and time-heterogeneous (user behavior changes over time); please refer to Section~\ref{sec:exploratory} for examples.
Unfortunately, mining non-stationary, transient, and time-heterogeneous stochastic processes is a challenging task.
It would be easier if trajectories were stationary (behavior is independent of wall clock times), ergodic (visits are infinitely repeated), and time-homogeneous (behavior does not change over time).
%Moreover, there are constraints when accessing items. 
%A check-in at a restaurant at a given location constraints the possible locations a user can be one hour later.
%Listening to an artist at Last.fm will constraint the navigation options users have available.

\begin{figure}[t]
    \centering
    \includegraphics[scale=.9]{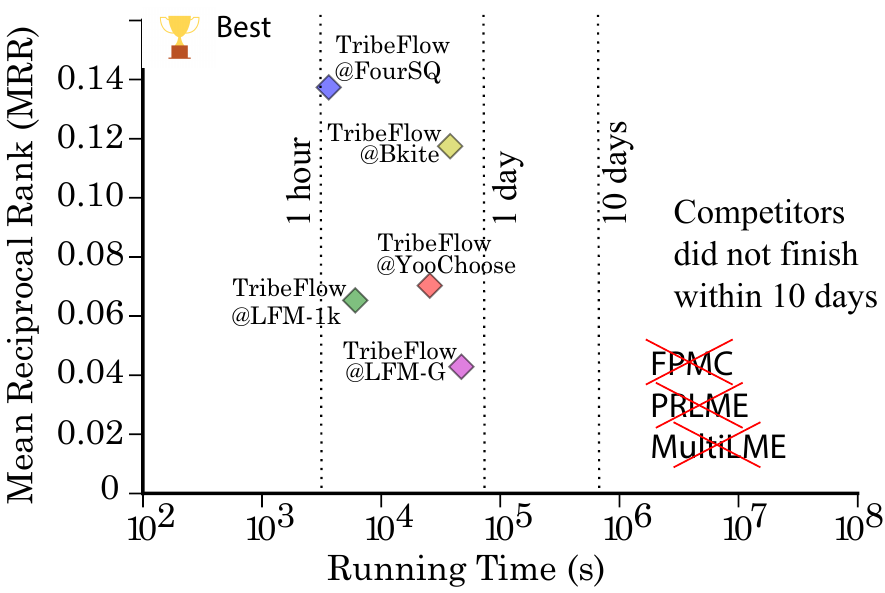}
    \vspace{-10pt}
    \caption{\model is at least an order of magnitude faster than state-of-the-art methods for next-item predictions.}
    \label{fig:cjewel}
    \vspace{-10pt}
\end{figure}

In this work we propose \model to tackle the problem of mining and predicting user trajectories.
\model takes as input a set of users and a sequence items they visit (user trajectories), 
including the timestamps of these visits if available, and outputs a model for personalized next-item  prediction (or next $n > 1$ items). % of new or existing users.
\model can be readily applied to personalized trajectories from next check-in recommendations, to next song recommendations, to product recommendations.
%\model has no tunable parameters of importance, being parameter-free in all of our results.
\model is highly parallel and nearly two orders of magnitude faster than the top state-of-the-art competitors.
In order to be application-agnostic we ignore application-specific user and item features, including time-of-day effects, but these can be trivially incorporated into \model.

To illustrate the performance of \model consider Figure~\ref{fig:cjewel},
where we seek to compare the Mean Reciprocal Rank (MRR) of \model over datasets with up to 1.6 million items and 86 million item visits (further details about this dataset is given in Section~\ref{sec:results}) against that of state-of-the-art methods such as Multi-core Latent Markov Embedding (MultiLME)~\cite{Moore2013}, personalized ranking LME (PRLME)~\cite{Feng2015}, and Context-aware Ranking with Factorizing Personalized Markov Chains~\cite{Rendle2010a} (FPMC).
Unfortunately, MultiLME, PRLME, and FPMC cannot finish any of these tasks in less than 10 days while for \model it takes between one and thirteen hours.
In significantly sub-sampled versions of the same datasets we find that \model is at least 23\% more accurate than its competitors.%, in fact, at least 23\% more accurate.

\model works by decomposing potentially non-stationary, transient,  time-heterogeneous user trajectories into {\em very} short sequences of random walks on latent environments that are stationary, ergodic, and time-homogeneous.
An intuitive way to understand \model is as follows. 
Random walks have been widely used for ranking items on {\em observed graph} topologies (e.g.\ PageRank-inspired approaches~\cite{Backstrom2011,haveliwala2002topic,Haveliwala2003,Jeh2003,page1999pagerank,richardson2001intelligent,Tong2006}); meanwhile,  
overlapping community detection algorithms~\cite{Airoldi2009,Krivitsky2009,Pool2014,Yang2013} also use {\em observed graphs} to infer latent weighted subgraphs.
But {\em what if we were not given the environments (weighted graphs \& time scales) but could see the output of a random surfer over them?}
\model sees user trajectories and infers a set of latent environments (weighted item-item graphs and their time scales) that best describe user trajectories through short random walks over these environments; after the short walk users perform a weighted jump between environments; the jump allows \model to infer user preference of latent environments.
Once the \model model infers the relationships between short trajectories and latent environments, 
we can give it any user history and current trajectory to infer a posterior over the latent environment that the user is currently surfing and, this way, perform accurate personalized next-item prediction using a random walk.
Our main contributions can be summarized as:

\newcommand{\gcheck}{\color{OliveGreen}{\bf \checkmark}}
\setlength{\tabcolsep}{2pt}

\begin{table*}[t!!]
\scriptsize
\centering
\caption{Comparison of Properties of State-of-art Methods.}
\begin{tabular}{l*{10}{c}}
\toprule
& MC MLE~\cite{Levin2008} & Gravity Model~\cite{Silva2006} & LDA/TM-LDA~\cite{Griffiths02,Wang2012} & LME/MultiLME~\cite{Chen2012,Moore2013} & P(R)LME~\cite{Feng2015,Wu2013} & FPMC~\cite{Rendle2010a} & Temporal Tensors & {\bf \model} \\
&&&&&&&&{\bf (our method)} \\
\cmidrule(r){2-9}
General Approach & \gcheck &   & \gcheck & \gcheck & \gcheck &  \gcheck & \gcheck&  \gcheck \\
Trajectory Model & \gcheck &   & & \gcheck & \gcheck & \gcheck & &  \gcheck \\
Personalized & &  & \gcheck & & \gcheck & \gcheck & \gcheck&  \gcheck \\ 
Multiple Time Scales  &  &  &  & &  &  & &  \gcheck \\ 
Trajectory Memory & \gcheck &  &  & & & & &  \gcheck \\ 
Sense Making &\gcheck & \gcheck  & \gcheck & & & & \gcheck&  \gcheck \\
Sub-Quadratic &\gcheck & \gcheck  & \gcheck & & \gcheck & \gcheck&  \gcheck & \gcheck \\
Scalable &\gcheck & \gcheck  & \gcheck & &  &  & \gcheck&  \gcheck \\ 
\bottomrule
\end{tabular}
\label{tab:salesmat}
%\vspace{-8pt}
\end{table*}

%\subsubsection*{Our Contributions}
\begin{itemize}[itemsep=-0.1mm,leftmargin=*]
\item {\bf (Accuracy).} In our datasets \model predictions are always more accurate than state-of-the-art methods.
The state-of-the-art methods include Latent Markov Embedding (LME) of Chen et al.~\cite{Chen2012}, Multi-LME of Moore et al.~\cite{Moore2013}, PRLME of Feng et al.~\cite{Feng2015}, 
and FPMC of Rendle et al.~\cite{Rendle2010a}. 
\model is also more accurate than an application of the time-varying latent factorization method (TM-LDA) of Wang et al.~\cite{Wang2012}. 
We also see why \model can better capture the latent space of user trajectories than state-of-the-art tensor decomposition methods~\cite{Matsubara2012a}.
\item {\bf (Parameter-free).} In all our results \model is used without parameter tuning. Because \model is a nonparametric hierarchical Bayesian method, it has a few parameters that are set as small constants and do not seem significantly affect the performance if there is enough data. 
The only parameter that affects performance ($B \in \mathbb{Z}^+$ explained in the model description) is safely set to a small constant ($B=1$).% in all of our results.
\item {\bf (Scalability).} \model uses a scalable parallel algorithm to infer model parameters from the data.
\model is between 48$\times$ to 413$\times$ faster than our top competitors, including LME, PLME, PRLME, and FPMC.
When available, we evaluated the performance of \model over the datasets originally used to develop the competing methods.
\item {\bf (Novelty).} \model provides a general framework ({\em random surfer over infinite latent environments}) to build upon for application-specific  recommendation systems.
%The general concept of \model (illustrated in Figure~\ref{f:alice}) is described in detail in Section~\ref{sec:model}.
%
% designed in a way that it works well even with very sparse data. This is useful when there are either ``too many items and too few users'' or ``too many users and too few items''. In particular, the Bayesian approach of \model enables us to find of patterns even in sparse datasets with millions of items and users, avoiding problems associated with maximum likelihood point estimates of Markov chain transitions. 
%
%\ju{These bullets seem repetitive in comparison with prior paragraph. I suggest to change the paragraph and keep only the bullets! Possibly also refer to prior work in more general terms here (state-of-the -art methods [x,y,z]. This will help the text to flow better. As is, it is quite dense!}
%\item {\bf (Trajectory Hypothesis Tests).} \model is also used to provide a robust method to test user trajectory hypotheses. People are diverse and therefore we expect that a trajectory hypothesis that matches well the behavior of a set of users will match poorly another set of users. Unlike HypTrails~\cite{x},  we show that \model is able to find the correct trajectory hypothesis even when there are multiple sets of users following distinct trajectory hypotheses. 
\end{itemize}
%\pagebreak
{\bf Reproducibility.} Of separate interest is the reproducibility of our work. Datasets, the \model source code and extra source code of competing methods that were not publicly available (implemented by us) can be found on our website\footnote{\url{http://flaviovdf.github.io/tribeflow}}.
\begin{figure}[tb]
\centering
\vspace{-2pt}
\includegraphics[width=3in,height=1.5in]{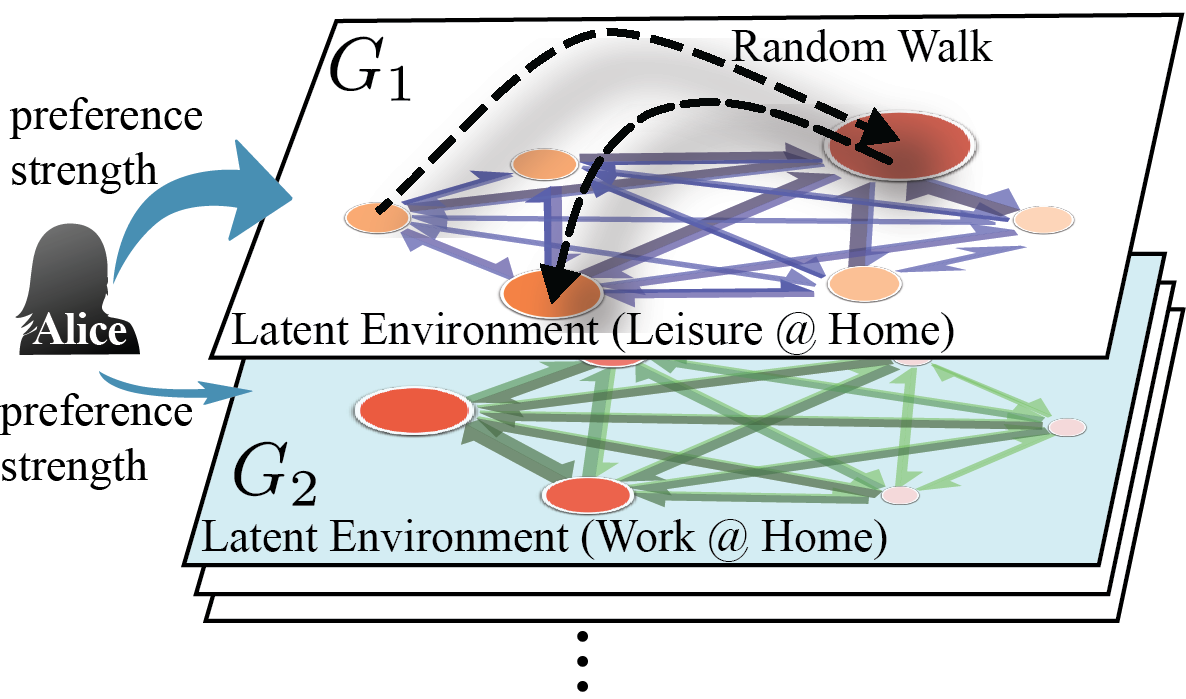}
\vspace{-5pt}
\caption{(\model) Alice randomly chooses a (latent) environment (weighted item-item graph $G_1$ and associated time scales) according to her preferences and surfs for a short time (two steps) before randomly jumping environments.\vspace{-7pt}}
% \model dynamically infers $G_1,G_2,\ldots$ using Bayesian nonparametric inference over user trajectory data, automatically adjusting the number of environments to the amount of data. \vspace{-15pt}}
\label{f:alice}
\end{figure}
%
%
%\subsubsection*{Outline}

We now present the outline of this work.
Section~\ref{sec:related} reviews the related work. 
Section~\ref{sec:model} describes the \model model.
Section~\ref{sec:results} presents our results on both small and reasonably-sized datasets of real-world user trajectories.
Section~\ref{sec:results}  also compares \model both in terms of accuracy and speed against state-of-the-art and naive methods.
Section~\ref{sec:exploratory} shows that \model has excellent sense-making capabilities.
Finally, Section~\ref{sec:conclusions} presents our conclusions.

\section{Related Work}
\label{sec:related}
%!TEX root = paper.tex
\newcommand{\myparagraph}[1]{{{\bf #1:}}}

How do we find a compact representation of the state space of network trajectories that allows easy-to-infer and accurate models? In this section, we present an overview of previous efforts related to \model that were also motivated by this question.

\myparagraph{Random Walks over Observed Networks}
The naive solution to the above question would be to simplify the state space of trajectories by merging nodes into communities via community detection algorithms~\cite{Rosvall2008,Yang2013,Pool2014,Airoldi2009,Rodriguez2014}.
However, we do not have the underlying network, only the user trajectories.
%Rodriguez et al.~\cite{Rodriguez2014} infers the underlying graph topology from epidemics but assumes a single underlying graph.
\model is able to infer latent environments (weighted graphs and inter-event times) from user trajectories without knowledge of the underlying network.

\myparagraph{Latent Markov Embedding}
Latent Markov Embedding (LME) \cite{Chen2012,Chen2013,Feng2015,Turnbull:2014,Wu2013} was recently proposed to tractably learn trajectories through a Markov chain whose states are projected user trajectories into Euclidean space. 
However, even with parallel optimizations~\cite{Chen2013} the method does not yet scaled beyond hundreds of thousands of users and items (as shown in Section~\ref{subsec:markov}).
The LME method can be seen as one of the first practical approaches in the literature to jointly learn user memory \& item preferences in trajectories.
Wu et al.~\cite{Wu2013} present a factorization embedding adding user personalization to the LME predictions, called PLME, which is not parallel and suffers from a quadratic runtime. Very recently Feng et al.~\cite{Feng2015} relaxed some of the assumptions of PLME in order to focus only on rankings (thus, P{\em R}LME), making the quadratic learning algorithm become linear but not multi-core as Multi-LME. 

\myparagraph{Factorizing Personalized Markov Chains} 
Rendle et al.~\cite{Rendle2010a} proposes Factorizing Personalized Markov Chains (FPMC) to learn the stochastic process of user purchasing trajectories.
FPMC predicts which item a customer would add next in his or her basket by reducing the state space of trajectories of each user into an unordered set (i.e., trajectory sequence is neglected).
The resulting Markov model of each user transition in the space of items forms the slice of a tri-mode tensor which then is embedded into a lower dimensional space via a tensor decomposition.
Similarly, Aizenberg et al.~\cite{Aizenberg2012} performs a matrix factorization for embedding.
Wang et al.~\cite{Wang2014} and Feng et al.~\cite{Feng2015} also make personalized factorization-style projections.
%Song et al.~\cite{Song2010} and Hsu et al.~\cite{Hsu2012} use singular value decomposition to embed the emission distribution of a HMM into a reproducing kernel Hilbert space, bypassing the problem of learning hidden states and the trajectory emission distribution.
%The work of Zheleva et al.~\cite{Zheleva2010} considers a highly application-dependent graphical model to explain user song plays through hidden taste and mood of songs but does not consider trajectories.
FPMC seems to be the most widely used method of this class that predicts next-items from unordered sets.

\myparagraph{Collaborative Filtering Methods}
Collaborative filtering methods can be broadly classified into two general approaches: memory-based (e.g.~\cite{Shi2014}) and item-based (e.g.~\cite{Salakhutdinov2008,Ma2008}).
In memory-based models next item predictions are derived from the trajectory each user independently of other users.
In item-based models, next item predictions are based on the unordered trajectories of all users, which disregards the sequence of events in the trajectory. 
More general unordered set predictions use Collective Matrix Factorization~\cite{Singh2008}.
% and Liu et al.~\cite{Liu2014}
Recently, Hierarchical Poisson Factorization of Gopalan et al.~\cite{Gopalan2015} deal with a different problem: item recommendation based on user ratings rather than user trajectories.
Chaney et al.~\cite{AllisonJ.B.Chaney2015} extends the Gopalan et al.\ model for recommendations when network side information exists but also does not consider trajectories.
The work of Wang et al.~\cite{Wang2012} uses a Latent Dirichlet Allocation-type embedding to capture latent topic transitions which we adapt to model trajectories in our evaluation (TM-LDA).
%Wang et al.~\cite{Wang2006} is one of the earliest works using the output of memory-based and item-based models as features of a classifier to jointly predict the next item in user trajectories.
%Similarly, in 
%consider specialized algorithms for recommending book lists using latent factors in a way similar to LME but still for predictions over unordered sets.
%

\myparagraph{Naive Methods}
Naive methods such as Gravity Model (GM)~\cite{Silva2006} are used to measure the average flow of users between items. Recently, Smith et al.~\cite{Smith2013} employs GMs to understand the flow of humans within a city. Galivanes et al.~\cite{Garcia-Gavilanes2014} employs GMs to understand Twitter user interactions. Note that GMs are application-dependent as they rely on a pre-defined distance function between items. 
Section~\ref{subsec:gravity} shows that \model is significantly more accurate than GM while retaining fast running times.

\myparagraph{Other Markov Chain Approaches}
A naive Markov Chain (MC) of the item-item transitions can be inferred via Maximum Likelihood Estimation (MLE)~\cite{Levin2008} but it does not provide enough flexibility to build good predictive models.
%Markov Switching Models (MSMs) of Frohwirth et al.~\cite{Frohwirth-Schnatter2008} consider a different problem than ours: to identify time series in multiple speakers and objects speech and video application tasks.
Fox et al.~\cite{Fox2010,Fox2014} and Matsubara et al.~\cite{Matsubara2014} propose Bayesian generalizations of Markov Switching Models~\cite{Frohwirth-Schnatter2008} (MSMs) for a different problem: to segment video and audio.
Liebman et al.~\cite{Liebman2015} uses reinforcement learning (RL) to predict song playlists but the computational complexity of the method is prohibitive even for our smallest datasets.

Hidden Markov Models (HMMs) can also be used to model trajectories. However, even in recent approaches~\cite{Wang2013,Toutanova2008,Goldwater2007}, fitting HMMs requires quadratic time in the number of latent spaces. The \model fitting algorithm is, in contrast, linear. HMMs are nevertheless interesting since inference is conditioned on the full sequence of a trajectory. \model can mimick this behavior with the $B$ parameter. We also considered novel HMM based model (called Stages) that has been proposed by Yang et al.~\cite{Yang2014}. Stages has sub-quadratic runtimes in the number of visited items (transitions) but the author-supplied source code did not converge to usable parameter values in our larger datasets. It is unclear why Stages is unable to converge over large datasets. In smaller datasets, where convergence did occur, Stages is less accurate than \model. The lower accuracy is likely because Stages is more focused on explicit trajectory commonalities and does not model personalized user transitions explicitly as \model does.

Table~\ref{tab:salesmat} compares \model with the strongest competitors in the literature for trajectory prediction and sense-making. \model is the only method that meets all criteria: general, personalized, multiple time scales, and scalable. Our inference algorithm is sub-quadratic in asymptotic runtime, as well as fully parallel. \model is the only approach that is accurate, general, and scalable.

\section{The \MakeUppercase{\model} Model}
\label{sec:model}
%!TEX root = paper.tex
\model models each user as a random surfer over latent environments.
User trajectories are the outcome of a combination of latent user preferences and the latent environment that users are exposed to in their navigation.
We use a nonparametric model of short user trajectory sequences as steps of semi-Markov random walks over latent environments composed of random graphs and associated inter-event time distributions.
Inter-event times are defined as the time difference between two consecutive items visited by a user.
The model is learned via Bayesian inference.
Our model is illustrated in Figure~\ref{f:alice}; in our illustration user Alice jumps to a latent environment ($\cM=1$) according to her environment preferences and performs two random walk steps on the graph $G_\cM$ with inter-event times associated to environment $\cM$. 
After the two steps Alice randomly jumps to another latent environment of her preference.

Random walks on our latent environments are not to be confused with random walks on dynamic graphs. 
In the former the underlying graph topology and associated environment characteristics do not change once they are drawn from an unknown probability distribution while in the latter the graph structure is redrawn at each time step.
In our applications probabilistic static environments seem like a better framework: a user with a given latent intent in mind (listen to heavy metal, eat spicy Indian food) at a given location (a webpage, a neighborhood) has item preferences (edge weights to songs, restaurants) similar to other like-minded users in the same network location.
A random graph framework would be more accurate if webpages and restaurants randomly changed every time users made a choice.
Our probabilistic environments are different from random graphs as the environment distribution is unknown and defines other characteristics such as the inter-event time distribution (simpler unidimensional examples of Markovian random walks on random environments with known distributions are given in Alexander et al.~\cite{alexander1981excitation} and Hughes~\cite[Chapter~6]{Hughes1995}.).

By construction, \model's semi-Markov random walk generates ergodic, stationary, and time-homo\-geneous (E-S-Ho) trajectories.
\model models potentially non-ergodic, transient, and time-hetero\-geneous user trajectories (Ne-T-He) as short sequences of E-S-Ho trajectories.
By its nature E-S-Ho processes are generally easier to predict than Ne-T-He processes.
And while a single user may not spend much time performing E-S-Ho transitions, other users will use the same E-S-Ho process which should allow us to infer well its characteristics.

%\newpage
\subsection{Detailed \model Description} \label{s:details}
The set of users ($\cU$) in \model can be agents such as people, bots, and cars.
The set of items ($\Omega$) can be anything: products, places, services, or websites. 
The latent environment  $\cM = 1,2,\ldots$ is a latent weighted clique $G_\cM = (\Omega,E_\cM)$ over the set of items $\Omega$. 
Edge weights in $E_\cM$ have gamma distribution priors, $w_{(\cdot,v)} = w_v \sim \text{Gamma}(\beta,1)$, $\forall v \in \Omega$.
In what follows we define the operator $|\cdot|$ to be the size of a set and $\otimes$ denotes the outer-product.

Each user $u \in \cU$ generates a ``sequence trajectory'' of length $B+1$ at the $t$-th visit, $t \geq 1$, $$(x_{u,t},\ldots, x_{u,t+B}) \in  \Omega^{B+1} \, , B \geq 1 \, , $$
before jumping to another environment $\cM'$ according to user preference distribution $\pi_{\cM' | u}$.
The entire trajectory $x_{u,1}, x_{u,2},\ldots$ of user $u \in U$ is the concatenation of such sequences.
%Whether the trajectory of a user can be segmented into multiples of $B$ is inconsequential to our inference problem as we infer posteriors over sliding windows. 

The time between observations $x_{u,t+k}$ and $x_{u,t+k+1}$ is the $k$-th inter-event time $\tau_{u,t+k}$, $k = 0,\ldots,B$.
Special care must be taken with the last event of a user, which does not have an inter-event time.
The random walk over $G_\cM$ is modeled as a semi-Markov process with inter-event time $\tau_{u,t} \sim \lambda(\cM)$  (a.k.a.\ holding or residence times).
Note again that inter-event times depend on the current latent environment.

\begin{figure}[t]
\centering
\includegraphics[width=2.4in,height=1.6in]{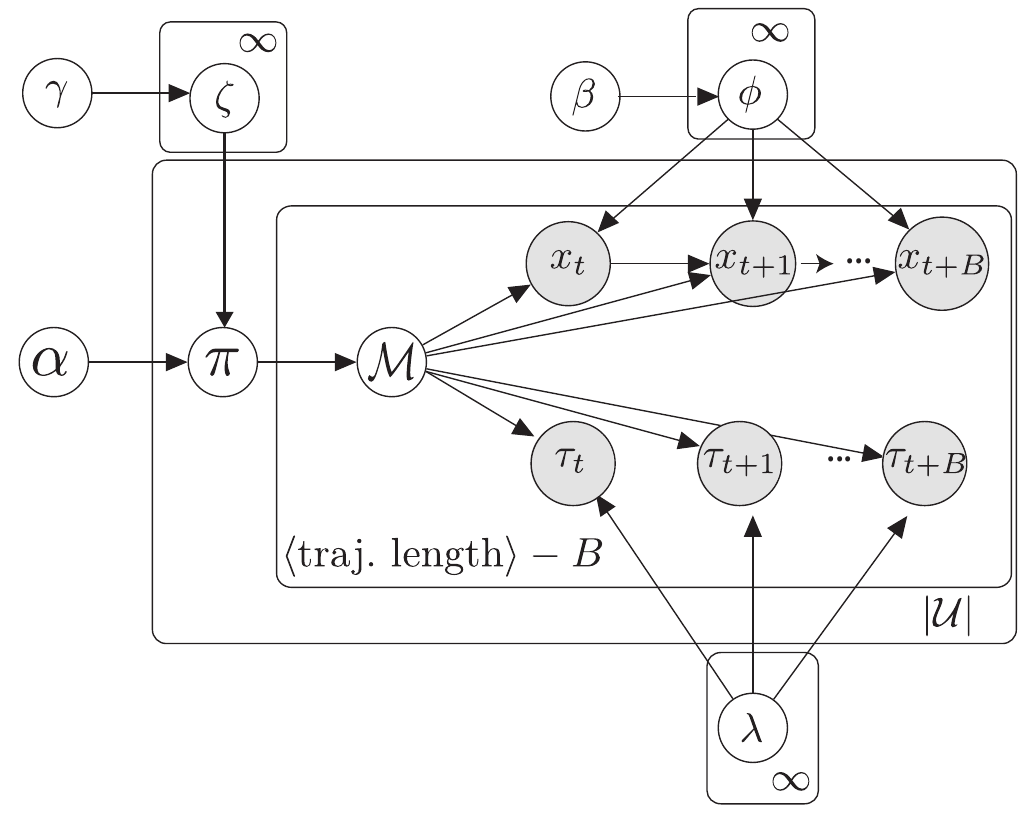}
\vspace{-5pt}
\caption{The \model model of the user and the semi-Markov transition probability matrix mixture.}
\label{fig:latentmodel}
\vspace{-3pt}
\end{figure}
%\begin{figure}[t]
%\centering
%\includegraphics[width=2.6in,height=0.7in]{figs/stochastic_complementation.pdf}
%\caption{Stochastic complement trick to treat item revisits.}
%\label{f:sc}
%\end{figure}
%
We now define the transition probability matrix of our random walk over the random graph $G_\cM$.
\begin{theorem} \label{t:P}
The random walk over graph $G_\cM$ of environment $\cM$ is a semi-Markov chain with a random $|\Omega| \times |\Omega|$ transition 
probability matrix distributed as
\begin{align}\label{e:P}
{\bf P}_\cM \sim (I - \text{diag}(\phi_\cM))^{-1} (\phi_\cM \otimes \phi_\cM - \text{diag}(\phi_\cM^2)) \, ,
\end{align}
where $\text{diag}(\cdot)$ is a diagonal matrix and 
$\phi_{\cM} \sim Dirichlet(\cdot \mid \beta)$.
Semi-Markov chain $\cM$ is stationary, ergodic, and time-homogeneous with high probability if $|\Omega| > 2$.
\end{theorem}
\begin{proof}
Without loss of generality we assume that the walker starts at $o \in \Omega$. 
The semi-Markov random walk with transition probability matrix ${\bf P}_\cM$ over $G_\cM$, $\cM \geq 1$, sees edge weights $w_v \sim \text{Gamma}(\beta,1)$, $\forall v \in \Omega \backslash \{o\}$.
The probability that the walk moves to $v \neq o$ is $w_v / S_{\neq u}$, where $S_{\neq o} = \sum_{j  \in \Omega \backslash \{o\}} w_j$.
Let ${\bf P}_\cM(o,\Omega \backslash\{ o\})$ denote the off-diagonal elements of the random walk transition probability matrix ${\bf P}_\cM$.
Because $\{w_j\}_{j \in \Omega}$ are independent and Gamma distributed then $${\bf P}_\cM(o,\Omega \backslash\{ o\}) = (w_v / S_{\neq u})_{v \in \Omega \backslash \{o\}}$$ follows a Dirichlet distribution $\text{Dirichlet}(\beta,\ldots, \beta)$~\cite[pp.\ 91]{kingman1992poisson}.
Note that ${\bf P}_\cM(o,o) = 0$.
A little algebra gives Eq.~\eqref{e:P}.
The chain is trivially stationary and time-homogeneous as transition probabilities do not change over time.
We now show the chain is ergodic.
Any state $j \in \Omega$ is reachable from any state $i \in \Omega$ as $({\bf P}_\cM)^n(i,j) > 0$ for some $n \geq 1$ implying that the chain is recurrent~\cite[Theorem~5.1, pp.\ 66]{karlin1969stochastic}. As $|\Omega| < \infty$ the chain is positive recurrent.
By construction $G_1$ is connected and thus ${\bf P}_\cM$ is irreducible.
For $|\Omega| > 2$ the graph $G_1$ is not bipartite and making the chain aperiodic.
If a chain is irreducible, aperiodic, and positive recurrent, then it is ergodic~\cite[pp.\ 85]{karlin1969stochastic}.
\end{proof}
\vspace{-2pt}
Note that the random walk does not have {\bf revisits}, $x_t \neq x_{t+1}$, $\forall t > 0$.
In the datasets we remove all revisits because re-consumption (repeated accesses to the same item) tends to be easy to predict~\cite{Figueiredo2014}, highly application-specific, and can be {\em decoupled entirely} from our problem via stochastic complementation using phase-type Markov chains with a single entry state such as the ones in Neuts~\cite{Neuts1979}, Robert and Le~Boudec~\cite{Robert1996} and Kleinberg~\cite{Kleinberg2003}.

Gathering all elements together we obtain the model illustrated in Figure~\ref{fig:latentmodel}, which can be seen as a random surfer taking $B$ steps over a latent graph $G_\cM$ and then randomly moving to a new environment according to the following generative model:
\vspace{-.1cm}
\begin{enumerate*}
 \setlength{\itemsep}{1pt}
 \setlength{\parskip}{0pt}
% \vspace{-5pt}
          \item  Draw $\zeta \sim \text{GEM}(\gamma)$ according to a stick-breaking process.
       \item For each user $u \in \cU$ sample $\pi_{\cM|u} \sim \text{Dirichlet}(\cdot \mid \alpha \zeta)$.
        \item Draw a semi-Markov random walk transition probability matrix ${\bf P}_\cM \sim (I - \text{diag}(\phi_\cM))^{-1} (\phi_\cM \otimes \phi_\cM - \text{diag}(\phi_\cM^2)) \, ,$ where $\text{diag}(\cdot)$ is a diagonal matrix and  $\phi_{\cM} \sim \text{Dirichlet}(\cdot \mid \beta)$, $\cM = 1,2,\ldots$.
       \item For a given user $u$ each sequence burst $(x_t,\ldots,x_{t+B})_{u}$ with inter-event times $(\tau_{t},\ldots,\tau_{t+B-1})$ is generated as follows:
    \begin{enumerate*}
 \setlength{\itemsep}{1pt}
 \setlength{\parskip}{0pt}
 \vspace{-5pt}
        \item Draw a latent semi-Markov chain\\ $\cM \sim \text{Multinomial}(\pi_{\cM|u})$.
        \item For $k=0,\ldots,(B-1)$ select item $x_{t+k}$ according to probability ${\bf P}_\cM(x_{t+k},x_{t+k+1})$ and inter-event time $\tau_{t+k'} \sim \lambda(\cM)$, $k'=0,\ldots,B$, where $\lambda(\cM)$ is the inter-event time distribution of environment $\cM$. Item $x_1$ is drawn uniformly from $\Omega$.
    \end{enumerate*}
\end{enumerate*}\vspace{-.1cm}

\subsection{Inferring \model Model from Data}

In what follows we describe how we learn \model from data. %tackle the problem of ``reversing'' the generative process to learn the latent environments and user preferences that when coupled with a random walk best correspond to user trajectories.
Given a set of user trajectories $\{(x_{u,1},x_{u,2},\ldots) : \forall u \in \cU\}$ from a set of items $x_{u,t} \in \Omega$, $t \geq 1$, we infer:
\begin{itemize}[itemsep=0mm,leftmargin=*]
\item The number of environments $K > 1$ from the data.
\item $K$ semi-Markov transition probability matrices $\{{\bf P}_\cM : \cM = 1,\dots,K\}$ corresponding random walks over a {\em finite} set of graphs $\{G_\cM$ : $\cM=1,\dots,K\}$.
\item A distribution of user environment preferences $\{\pi_{\cM|u} : u \in \cU\}$.
\end{itemize}
If the inter-event time distribution $\lambda(\cM)$ comes of a known family we can also get a distribution of inter-event times for each environment. %However, we found that our heuristics that do not entail a parametric posterior distribution over inter-event times performs better than normally distributed models.
The probability that a user sees a sequence $x_{u,t},\ldots,x_{u,t+B}$ with inter-arrival times $\tau_{u,t},\ldots,\tau_{u,t+B}$ at environment $\cM$ is 
\begin{align*} 
    P[&x_{u,t},\ldots,x_{u,t+B},\tau_{u,t},\ldots,\tau_{u,t+B} | \cM] = \nonumber\\
     &\prod_{k=0}^{B-1} {\bf P}_\cM(x_{u,t+k},x_{u,t+k+1}) P[\tau_{u,t+k} | \cM] P[\tau_{u,t+B} | \cM] \, \nonumber,
\end{align*}
with ${\bf P}_\cM$ as given in Eq.~\eqref{e:P}.
The probability we observe such burst for user $u \in \cU$ is then
\begin{equation}
\label{eq:Bm}
\begin{aligned}
&P[x_{u,t},\ldots,x_{u,t+B},\tau_{u,t},\ldots,\tau_{u,t+B} | u] = \\ 
  &\sum_{K =1}^\infty P[ \zeta ] P[\tau_{u,t+B} | \cM]  \pi_{\cM|\alpha\zeta,u} \\
  &\qquad \times \prod_{k=0}^{B-1} {\bf P}_\cM(x_{u,t+k},x_{u,t+k+1}) P[\tau_{u,t+k} | \cM]  \,  \, ,
\end{aligned} 
\end{equation}
where $P[\zeta]$ is the stick-breaking prior over $\cM$.
Unrolling Eq.~\eqref{e:P} into Eq.~\eqref{eq:Bm} we obtain the equation that describes the trajectory:
\begin{equation}
\label{eq:full}
\begin{aligned} 
    P[&x_{u,t},\ldots,x_{u,t+B},\tau_{u,t},\ldots,\tau_{u,t+B-1} | u] \propto \\
      &\sum_{K =1}^\infty P[ \zeta ] P[\tau_{u,t+B} | \cM]\,  \pi_{\cM|\alpha\zeta,u} \\
       &\qquad \times \prod_{k=0}^{B-1} \frac{\phi_\cM(x_{u,t+k+1})}{1 - \phi_\cM(x_{u,t+k})} P[\tau_{u,t+k} | \cM]   \, .
\end{aligned}
\end{equation}

We use collapsed Gibbs sampling to estimate the model parameters. Initially, given a sequence of size $B + 1$, we transform user trajectories into a set $\cD$ of tuples using a sliding window over the trajectories of each user. 
To exemplify, for $B=2$ each entry is:
$$(u, x_{u,t},x_{u,t+1},x_{u,t+2},\tau_{u,t},\tau_{u,t+1}) \in \cD \:, t \geq 1 .$$
This tuple represents the user $u$, the trajectory $x_{u,t},\ldots,x_{u,t+B}$ and the inter-event times $\tau_{u,t},\ldots,\tau_{u,t+B}$ for every time $t \geq 1$. 
The use of sliding windows, while not theoretically justified by our model, tremendously simplify our inference problem by not forcing us to decide how 
to segment the data into blocks of $B$ events or forcing us to make $B$ a random variable.
Adding random jumps would entail a costly forward-backward inference procedure needed to decide when users jump between environments.
To infer \model, our heuristic starts with an initial estimate of the number of environments $K$ and randomly assign each tuple in $\mathcal{D}$ to one environment.

After the initial assignment we count the number of tuples of each user $u$: $n_u = \sum_{\forall (u', \ldots) \in \mathcal{D}} {\bf 1}(u' = u)$, where ${\bf 1}$ is the indicator function. 
We also count the number of times environment $\cM$ is assigned to a tuple from user $u$: $e_{\cM,u}$, as well as the joint count of items, at any position, and environments:  $c_{i,\cM}$, and count the number of tuples assigned to an environment  $\cM$:  $a_{\cM}$.
Assuming, for now, that $\zeta$ is given, we can infer~\cite{Griffiths02}
\begin{align} \label{eq:dirichlet}
    \pi_{\cM|\alpha\zeta,u} = \frac{e_{\cM,u} + \alpha \zeta(\cM)}{n_{u} + K\alpha\zeta(\cM)} \quad , \quad
    \phi_{\cM}(i) = \frac{c_{i,\cM} + \beta}{a_{\cM} + |\Omega|\beta}.
\end{align}
%We note that $P[ \cM ]$ is proportional to $a_{\cM}$.
% \model model incorporates inter-event time information and a {\it dynamic expansion/contraction of the number of latent random environments} ($K$). That is, our full \model learns $\lambda(\cM)$ and $\zeta(\cM)$.
We then employ ECME~\cite{Gelman2013} inference where: (1) the e-step consists of one pass over the entire dataset performing a Gibbs sampling update (in other terms, one iteration of the collapsed Gibbs sampler); (2) an m-step where the algorithm considers the probability of inter-event times according to the following procedure. 

\begin{table*}[t!]
\centering
%\small
\caption{Summary of our datasets.}% For testing, we leave out the last 30\% of transitions on each trace (when ordered by time stamp). The exceptions for this rule are the Yes data which has a pre-defined test set.}
%\resizebox{1.65\columnwidth}{.6in}{
\begin{tabular}{lccccccc}
\toprule
& \Hentity & \SDentity & \# Transitions & \# \Hentity & \# \SDentity & Inter-event Times & Timestamp Span \\
\cmidrule(r){2-8}
Last.FM-1k          & User & Artist     & 10,132,959 & 992       & 348,156   & Yes & Feb. 2005 to May 2009   \\
Last.FM-Groups      & User & Artist     & 86,798,741 & 15,235    & 1,672,735 & Yes & Feb. 2005 to Aug. 2014  \\
\cmidrule(r){2-8}
BrightKite          & User & Venue      & 2,034,085  & 37,357    & 1,514,460 & Yes & Apr. 2008 to Oct. 2010  \\
FourSQ          & User & Venue      & 453,429    & 191,061   & 87,345    & Yes & Dec. 2012 to April 2014 \\
\cmidrule(r){2-8}
YooChoose           & Session & Product & 19,721,515 & 6,756,575 & 96,094    & Yes & Apr. 2014 to Sept. 2014 \\
\cmidrule(r){2-8}
Yes         & Playlist & Song   & 1,542,372  & 11,139    & 75,262    & No & (No timestamps) \\
\bottomrule
\end{tabular}
%}
\label{tab:data}
\end{table*}

If the inter-event times of environments $\cM=1,2,\ldots$ have a known probability law $\lambda(\cM)$ we can include this law in our model with the appropriate priors (e.g.\ a Normal distribution with fixed variance can have a Normal prior), updating the distribution parameters in each m-step. If the law of $\lambda(\cM)$ is unknown but inter-event times are observed, we use the empirical complementary cumulative distribution function (ECCDF) in the the following heuristic: We estimate the ECCDF of each $\cM$ based on the entries assigned to $\cM$ on the last e-step. 
If $T_\cM$ is the random variable that defines the inter-arrival at environment $\cM = 1,\ldots$.
For inter-event time $\tau_{u,t}$ the probability $P[T_\cM > \tau_{u,t}]$ given the current ECCDF is the number of entries whose observed inter-event times in $\cM$ are greater than $\tau_{u,t}$.
Thus, ${\bf 1}(T_\cM > \tau_{u,t})$ is a Bernoulli random variable~\cite{Wasserman2004} with parameter $p = P[T_\cM > \tau_{u,t}]$. 
Adding a conjugate prior $Beta(1, K-1)$ to the Bernoulli gives the following predictive posterior~\cite{Gelman2013}:
\begin{align} \label{eq:ptau}
    P[\tau_{u,t} | \cM] \propto \frac{b_{>\tau_{u,t},\cM}  + 1}{n_{\cM} + K} \, ,
\end{align}
where $b_{>\tau_{u,t},\cM} $ is the number of tuples currently assigned to $\cM$ that have inter-arrival times greater than $\tau_{u,t}$. It is easy to see from Eq.~\eqref{eq:ptau} that transitions with large inter-event times w.r.t.\ other inter-event times in $\cM$ are less likely to have been generated by $\cM$. This captures the intuition that inter-event times in an environment should be similar and is an integral part of how we learn $\zeta$ as detailed below. 
%The $Beta(1, K-1)$ prior represents a {\it non-informative} estimate of $P[\cM | \tau_{u,t}]$ if no additional information (e.g., $b_{>\tau_{u,t}}$) were provided. That is, if not information were provided, we assume that for any inter-event time there is a uniform distribution (1 out of every $K$)%However,  
%That is, if we had no additional information a uniform distribution (1 out of every $K$) would be assumed for $P[\tau_{u,t} | \cM]$. 
%
In our experiments the ECCDF heuristic consistently provides better results than assuming Normal inter-event time distributions. Testing other application-specific inter-event time posteriors is of interest in more application-specialized future work.

We now infer $K$, the number of environments, by adapting the partial expectation (partial-e) and partial maximization (partial-m) heuristics of Bryant and Sudderth~\cite{Bryant2012} for hierarchical Dirichlet processes (HDPs) as follows: (a) merge pairs of redundant environments when there is gain in the joint posterior probability of model \& data; and (b) split environments based on their inter-event times by separating the 5\% entries with highest inter-event time if there is a marginal gain in the posterior. The merge step (partial-e) is equivalent to the one described in~\cite{Bryant2012}. Our split step (partial-m) splits a environment with high variance inter-event times into two environments with lower variance if the splitting improves the joint posterior probability of model \& data.

{\bf Parallel Learning:} Finally, the model can be fully trained in parallel using the approach described in Asuncion et al.~\cite{Asuncion2011}. Each processor learns the model on a subset of the dataset. After a processor performs an ``E'' and an ``M'' steps, another processor is chosen at random to synchronize the model state (merge count matrices of the Gibbs sampler and update ECCDF estimates), leaving them with the same count matrices and ECCDF estimates. After a fixed number of iterations, which we fix as 200, every processor meets a barrier and the partial-e and partial-m steps (merges and splits) are performed on a single master processor. After these partial-e and partial-m steps, the master-processor updates all slave processors and the parallel learning continues as previously described. The learning ends after a fixed number of iterations, which we set as $2,\!000$ for the results in Section~\ref{sec:results}. 

{\bf Inference without Timestamps:} In a few of our datasets timestamps are not available. 
In these cases we only need to infer ${\bf P}_\cM$ and $\pi_{\cM|\alpha \zeta, u}$ without the need to take into account $P[\tau_{t+k} | \cM] $.
In our results we also take the opportunity to assume that the number of environments is fixed ($\zeta(\cM) = 1$) and infer the model posteriors using collapsed Gibbs sampling~\cite{Griffiths02,Matsubara2014,Harvey2011,Yin2015,Kang2013,Kang2013a,Wang2012} employing Eqs.~\eqref{eq:full} and~\eqref{eq:dirichlet}, updating the counts at each iteration (updating the posterior probabilities). We denote the latter simpler approach {\bf \model-NT} (\model-NoTimestamps).

{\bf Prior Parameters:} For all of our experiments, we fix prior parameters $\alpha = 50/K$ and $\beta = 0.001$. Constant $\gamma$ does not need to be stated explicitly. In the presence of large amounts of data these priors are expected to have little impact on the outcome of Bayesian nonparametric models~\cite{Gelman2013}.

\subsection{\model Predictions}

In this work \model has two prediction tasks: {\bf next-item} likelihood predictions and {\bf ranking}.
%In what follows we describe how \model can be used for prediction tasks.
%
%\paragraph{Next-item predictions} 
The personalized predictive likelihood of user $u \in \cU$ for candidate next-item $\tilde{x}_{u,t+1} \in \Omega$ is based on the last $B$ items and inter-event times of the user.
Observing that user $u$ has chosen items $x_{u,t-1},\ldots,x_{u,t-B}$ with inter-event times $\tau_{u,t-1},\ldots,\tau_{u,t-B}$ gives a posterior probability that $u$ is performing a random walk at latent environment $\cM \in \{1,\ldots,K\}$, where $K$ is the learned number of latent environments, typically $K < 10^3$ in our experiments.
More precisely, the posterior probability that $u\in \cU$ is in environment $\cM$ after choosing items $x_{u,t-1},$ $\ldots,x_{u,t-B}$ with inter-event times $\tau_{u,t-1},\ldots,\tau_{u,t-B}$ is $$P[\cM | u,  x_{u,t-1},\ldots,x_{u,t-B},\tau_{u,t-1},\ldots,\tau_{u,t-B}],$$ which yields the full likelihood
\begin{align} \label{eq:ll}
&P[\tilde{x}_{u,t+1} | u, x_{u,t-1},\ldots,x_{u,t-B},\tau_{u,t-1},\ldots,\tau_{u,t-B}] =   \\
%     & \sum_{\cM = 1}^K {\bf P}_\cM(x_{u,t},\tilde{x}_{u,t+1}) P[\cM | u,  x_{u,t-1},\ldots,x_{u,t-B}, \tau_{u,t-1},\ldots,\tau_{u,t-B}] = \nonumber\\
     &  \left( \prod_{k=1}^{B} {\bf P}_\cM(x_{u,t-k},x_{u,t-k+1}) P[\tau_{u,t-k} | \cM]  \pi_{\cM|\alpha \zeta,u} \right) \nonumber \\
     &  \times \frac{\sum_{\cM = 1}^K {\bf P}_\cM(x_{u,t},\tilde{x}_{u,t+1})}{\sum_{\cM = 1}^K \prod_{h=1}^{B} {\bf P}_\cM(x_{u,t-h},x_{u,t-h+1}) P[\tau_{u,t-h} | \cM]  \pi_{\cM|\alpha \zeta,u}} \nonumber .
\end{align}

%\paragraph{Ranking} 
The task of ranking the next items is easier and faster than predicting their likelihood.
This is because we can speed up the predictions by not computing the denominator in eq.~\eqref{eq:ll}, as the denominator is the same for all values of $\tilde{x}_{u,t+1} \in \Omega$.
Note that our rankings are personalized and consider the inter-event times.

\section{Results}
\label{sec:results}
%!TEX root = paper.tex
Previous sections introduce \model and explain how to infer its posteriors from data. 
We now turn our attention to compare \model against state-of-the-art approaches, solving the very same problems over some of the very same datasets (if publicly available) as their original papers. 
We also use larger publicly available dataset and one ever larger dataset that we collected for this study (available for download at our website).
In what follows Section~\ref{subsec:transitions} contrasts \model against state-of-the-art methods for next-item ranking. 
Section~\ref{subsec:markov} compares \model against methods that learn latent Markov chains and predict the likelihood of next items. 
Finally, Section~\ref{subsec:gravity} contrasts \model ability to predict average mean-field user flows against that of Gravity Model. 

Although \model is able to predict not just the next-item but also the next $n \geq 1$ items, our evaluations are based at next-item predictions because previous efforts mostly focused their evaluations on this task.
We also consider the reconsumption problem (consecutive visits to the same item) treated in Figueiredo et al.~\cite{Figueiredo2014} as a separate, often easier, problem that can be dealt with via stochastic complementation as discussed in Section~\ref{s:details}.

%On our final comparison, we discuss how \model compares with Tensor decomposition approaches which exploit a time dimension (Section~\ref{subsec:temptensors}) to extract meaningful patterns of user behavior.

\subsection{Datasets and Evaluation Setup} \label{subsec:data}
Our datasets encompass three broad range of applications: (a) location-based social networks (check-in datasets), (b) music streaming applications, and (c) user clicks on e-commerce websites. 
It is important to point our that recommendation engines and user interfaces influence user navigation and their trajectories. Such effects are considered to be an integral part of our predictive task.
But  \model can as easily learn pure user preferences if given a dataset with no environment bias (when that is possible). 

Table~\ref{tab:data} summarizes our datasets showing the number of users and items, the total number of $x_{u,t}, x_{u,t+1}$ pairs visited all users (or transitions/trajectories when $B=1$) as well as the time span covered by the dataset. The set of items, $\Omega$, can be songs or artists on music datasets, venues on check-in data, and products on e-commerce data. The set of users $\cU$ are individuals, a ``playlist'', or a ``browser session'' as described next.
\vspace{-0.5em}
\begin{description}[leftmargin=0in,labelindent=0in]
 \setlength{\itemsep}{4pt}
 \setlength{\parskip}{0pt}
 \setlength{\leftmargin}{0pt}
    \item [Last.FM-Groups.] 
     Last.FM is a music streaming service that aggregates data from various forms of digital music consumption, ranging from desktop/mobile media players to other streaming services. This dataset was crawled in August 2014, using the user groups\footnote{Pages in which the user discusses musical artists} feature from Last.FM. We manually selected 15 groups of pop artists and two general interest groups\footnote{{\it Active Users}, {\it Music Statistics}, {\it Britney Spears}, {\it The Strokes}, {\it Arctic Monkeys}, {\it Miley Cyrus}, {\it LMFAO}, {\it Katy Perry}, {\it Jay-Z}, {\it Kanye West}, {\it Lana Del Rey}, {\it Snoop Dogg}, {\it Madonna}, {\it Rihanna}, {\it Taylor Swift}, {\it Adelle}, and {\it The Beatles}}.  For each group, we crawled the listening history of a subset of the users (the first users listed in the group).%\footnote{We focus on the first (more active) users due to rate limits.}. %The total number of crawled users in 15,235. While this dataset has biases towards more active users and pop artists (given our choice of groups), it has over 10 times more users and plays than the Last.FM-1k dataset. Moreover, it allows us to analyze the behavior of fans (based on group membership) of different artists. Finally, this dataset also contains the self-declared age (at the time) of the users and the nationality.

\item[Last.FM-1k.] The second Last.FM dataset was collected in 2009 using snowball sampling by Celma et al.~\cite{Celma2010}. %After the snowball sampling, 992 uniformly random users were selected. %The dataset contains, for each user, the complete listening history (all plays) from February 2005 to May 2009, the self-declared nationality, age (at the time) and registration date~\cite{Celma2010}. 
%Our first Last.FM dataset, Last.FM-1k, is comprised of 10,132,959 information flows triggered by 992 users to 348,156 musical artists.

\item[BrightKite.] Brightkite is a location based social network (LBSN) where users share their current locations by check-ins. In this publicly available dataset, each \sdentity is a location where \hentity are checks-in. Collected by Cho et al.~\cite{Cho2011}. %The dataset accounts for roughly one and a half year of user check-ins. 
%The users account for 2,034,085 check-ins from April 2008 to October 2010\footnote{\url{http://snap.stanford.edu/data/loc-brightkite.html}}. 
%The Brightkite dataset was originally gathered by Cho {\it et al.}~\cite{Cho2011} and is available online\footnote{\url{http://snap.stanford.edu/data/loc-brightkite.html}}.

\item[FourSQ] Our second LBSN dataset was gathered from FourSquare by Sarwt {\it et al.}~\cite{Sarwat2014} in 2014. %Like BrightKite, this dataset is also focused on check-ins from users. %Similar to the BrightKite dataset, our FourSquare data captures the check-in behavior of 191,061 users on 87,345 venues. %Both our BrightKite and the FourSquare datasets contains the residence location (e.g., where users reside) for a subset of users.

\item[YooChoose.] This dataset is comprised of user clicks on a large e-commerce business. Each \hentity is captured by a session and the trajectories capture clicks on different products within the session. %The YooChoose captures roughly 20 million clicks, from 6 million sessions to approximately 20 thousand distinct products. %This dataset contains the residence time information for each click. 
As ``users'' are actually browser sessions, there is an upper-limit of 12 hours on recorded inter-event times of a single ``user'' (session).

\item[Yes.] Finally, the Yes dataset consists of song transitions (playlists) of popular broadcast (offline) radios in the United States. {\it This dataset does not provide explicit user information or timestamps}. However, we use \model-NT by defining the playlists as users and each song as an item. The Yes dataset was collected by Chen et al.~\cite{Chen2012,Chen2013} to develop LME.   %\ju{I suppose that \model-NT was clearly defined before this point. If this is a variation of \model, I think this should be made clear. }
%In total this dataset is comprised of 11 thousand playlists, 75 thousand songs from 1.5 million song transitions. It is important to note that the Yes dataset does not contain time-stamped transitions, thus we cannot account for residence times explicitly. Such a fact does not limit the use of \model which can be explored for both datasets with and without residence times. 
\end{description}
No filtering or trimming is done over the original data for the results shown in Figure~\ref{fig:cjewel}.
In our evaluation of trajectory predictions we divide the datasets into ``past'' ($\cD_\text{past}$) and ``future'' ($\cD_\text{future}$) by selecting a timestamp that splits the dataset into ``the first 70\% transitions'' for ``past'' (training set) and the remaining 30\% transitions in the data in ``future'' (test set), with the exception of Yes that has no timestamps.
This training and testing scenarios best represent real-life situations where the training is performed in batches over existing data.
For instance, in this realistic training and test setting we may need to predict the next transitions of a user that belongs to the test set but not to the training set. 
Note that some users will have trajectories confined in the ``past'' dataset while the trajectory of ``new users'' may be entirely placed in the ``future'' dataset.
For the Yes data the training and testing sets are the ones pre-defined by Chen et al.~\cite{Chen2012,Chen2013}.
Due to limitations in scalability of state-of-the-art methods we also exploit subsamples of our larger datasets when necessary. 
These subsamples have the first 1000, 2000, 5000, 10,000, 20,000 and 100,000 transitions ordered by timestamps. 
We test \model's robustness through these dataset subsamples.

{\bf Setup:} Our tests run on a server with 2$\times$10-core Intel Xeon-E5 processors and 256 GB of RAM.  
We tested and present results of \model with $B = 1,2,3,4,5$, fixing other hyper-parameters as discussed in Section~\ref{sec:model}.
% finding that larger $B$ improves accuracy at FourSQ, LFM-1k, LFM-Groups and reduces accuracy at Bkite and YooChoose (as explored later in this section). 
%The conclusion is that performance depends on the dataset and, to make our results parameter-free, we fix $B=1$ for all results presented in this work.
Our results show that regardless of the choice of $B$, \model outperforms competitors. % for all results presented in this work.

%\ju{Does this imply that we are not showing results with $B>1$?  I would rephrase this paragraph to clarify that.}
%\ju{It is confusing what \model-NT is. Why do you say only \model in paragraph Setup and explicitly mention \model-NT in the following one. This makes the reader wonder. Avoid it as much as possible. Try to say all variations of \model.}

\subsection{\model for Next-item Ranking} \label{subsec:transitions}

Before we compare \model against competing approaches it is important to note {\bf only \model can handle our larger datasets} as evidenced by Figure~\ref{fig:cjewel}.  
All of our comparison against competing state-of-the-art methods are {\bf performed over small or sub sampled datasets due to the large execution times.}

We start our discussion on ranking methods by focusing on the evaluation metric: the mean reciprocal rank.
The reciprocal rank, $RR$, is the inverse of the position of the destination ${x}_{u,t+1}$ on the ranking of {\it all} potential candidates in decreasing order. That is, if  candidate, $\tilde{x}_{u,t+1}$, destinations {\it Big Brewery}, {\it Pizza Place}, and {\it Sandwich Shop} are ranked with probabilities $0.4$, $0.5$, and $0.1$ respectively, and the {\it true} destination was {\it Big Brewery}, the reciprocal rank has value $1/2$. Using \model, the RR can be computed using Eq.~\eqref{eq:ll} as follows (with $B=1$): 
\begin{align}RR({x}_{u,t+1}, x_{u,t}, \tau_t, u) = \frac{1}{rank(P[{x}_{u,t+1} \mid x_{u,t}, \tau_{u,t}, u])}.\end{align} 
    \noindent Based on the reciprocal rank we define a single metric measured over the entire test set. This metric is simply the mean of reciprocal rank values over every transition in $\mathcal{D}_{test}$ (MRR).

%\begin{align}
%MRR(\cD_\text{future}) = \frac{\sum_{x_{u,t+1}, x_{u,t}, \tau_{u,t}, u \in \cD_\text{future}} RR(x_{u,t+1}, x_{u,t}, \tau_{u,t}, u)}{|\cD_\text{future}|}.
%\end{align}
%The mean reciprocal rank (MRR) has been recently used in Allison et al.~\cite{AllisonJ.B.Chaney2015}.
%\\

\myparagraph{\model with HDP heuristics} Before comparing \model with competing approaches, we test if our heuristic of HDP expansion and contraction of latent environments \& inter-event time ECCDF inference improves the quality of the predictions. That is, we run \model with and without the partial e- and m-steps described in Section~\ref{sec:model}. 
We train \model with an initial guess of $K=100$ environments and \model-NT (without partial e- and m-steps) with $100$ environments. 
We evaluate the MRR at our largest datasets: BrightKite, LastFM-Groups, LastFM-1k, and YooChoose, and one small dataset (FourSQ).
At FourSQ and YooChoose both \model and \model-NT have the same accuracy, possibly because of lower quality timestamps: FourSQ user data is sub-sampled {\em in time} and YooChoose browser sessions timeout at 12 hours. 
For the remaining datasets the MRR gains of \model over \model-NT are 7\%, 45\%, and 53\%, for LastFM-1k, LastFM-Groups, and BrightKite, respectively. 
This shows that our partial e- and m-step heuristics can significantly improve the results.
\begin{figure}[t]
    \centering
    \includegraphics{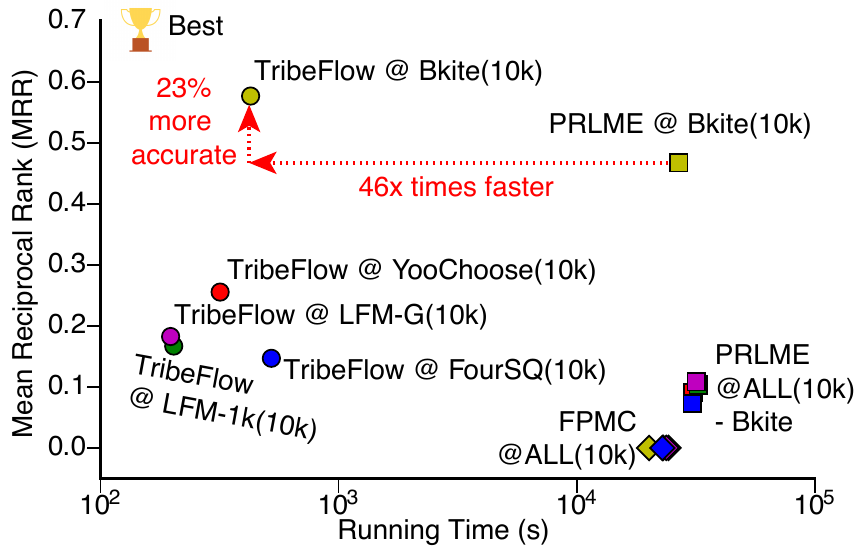}
    \vspace{-20pt}
    \caption{\model again outperforms in ranking task (subsampled datasets with only $\bf 10^4$ transitions).}
    \vspace{-5pt}
    \label{fig:mrr_comp}
\end{figure}

%Such gains come at a cost of a higher training time, as shown in the figur which is, nevertheless, tractable taking less than a day. On the other datasets, no significant losses in accuracy is noticeable.
%We have validated these results with a one sided Kolmogorov-Smirnov test on distributions of reciprocal rank values. This test showed that the improvements are statistically significant ($p < 0.001$). 
%From the figure we can see that \model, regardless of expansion/contraction, can handle very large datasets (such as LastFM-Groups) with a training time of under 1 day. More importantly, we can also see that in 3 cases (LastFM-Groups, LastFM-1k and BrightKite), the expansion/contraction brings significant gains from 7\% (LastFM-1k) up to 53\% (BrightKite). Such gains come at a cost of a higher training time, as shown in the figure, which is, nevertheless, tractable as we have discussed. On the other datasets, no significant losses in accuracy is noticeable. We have validated these results with a one sided Kolmogorov-Smirnov test on distributions of reciprocal rank values. This test showed that the improvements are statistically significant ($p < 0.001$). 

\myparagraph{\model v.s.\ state-of-the-art} We now turn to our comparison of \model against state-of-the-art competing ranking approaches. 
Our first competitor is  Factorizing Personalized Markov Chain (FPMC) of Rendle et al.~\cite{Rendle2010a}. 
FPMC was initially proposed to predict the next object a user will insert into an online shopping basket. 
If we consider each source $x_t$ as a size one shopping basket, FPMC can be used to rank the candidate items, or destinations $\tilde{x}_{t+1}$, that a user will consume next. 
Our second competitor is the best-performing Latent Markov Embedding (LME) method in our datasets: Personalized Ranking by Latent Markov Embedding (PRLME) of Feng et al.~\cite{Feng2015}. 
Inspired by Personalized Latent Markov Embedding~\cite{Wu2013} (PLME), the PRLME approach focuses on rankings and not on extracting Markov chains.% likely leading to the more accurate results we observed for PRLME. %As shown by the authors, their method is better than both MultiLME and PLME~\cite{Feng2015}.

Figure~\ref{fig:mrr_comp} presents the results of \model against FPMC and PRLME over datasets Bkite, FourSQ, LFM-Groups, LFM-1k, YooChoose subsampled to the first {\bf 10,000} transitions only due to scalability issues of FPMC and PRLME (even when limiting inference with 1000 stochastic gradient descent iterations). 
The figure shows the MRR scores (y-axis) against running times in log-scale (x-axis). 
Each point in the figure represents one method-dataset pair, while different datasets are represented by different colors. We label each point in order to help readability. 
The top-left corner of the figure indicates the best accuracy and shorter runtime. 
In these tests, as with all our tests, we use 70\% of the initial user transitions to perform inference and the last 30\% to test accuracy.

In Figure~\ref{fig:mrr_comp} we see that {\em \model is 23\% more accurate and 46$\times$ times faster than the best result of PRLME}.
In all cases \model is more accurate and faster than PRLME and FPMC, oftentimes \model is two orders of magnitude faster.
This is surprising as \model runs parameter-free while PRLME and FPMC both optimize over a large parameter space to obtain their best results\footnote{We perform a grid-search over parameters, testing 3 to 5 different values for each parameter of each method.}.
Our evaluation also considers a range of subsampled transitions in the datasets, from $10^3$ to $10^5$ transitions.
The accuracy results are similar in all cases, with \model showing consistently more accurate results.
Interestingly, using the results from the LastFM-Groups subsamples ($1000, 5000, 10000, 20000,10^5)$ using a simple linear regression reveals that it would take {\it over six years} to run PRMLE and FPMC in the full LastFM-Groups dataset using our server. 
The minimum  expected running time of PRMLE/FPMC on a complete dataset is 14 days for the FourSQ dataset (the smallest dataset).
The parameter search, as well as the lack of parallelism, greatly impact the runtime of PRLME and FPMC. Nevertheless, both methods are sub-quadratic as is \model. In the next sub-section we shall compare \model with another fully parallelized baseline.

Finally, we point that we also compare \model with the Stages method proposed by Yang et al.~\cite{Yang2014}. Our simulations used an author-supplied source code\footnote{\url{http://infolab.stanford.edu/~crucis}} that, unfortunately, did not converge to usable parameter values except over a few of the sub-sampled datasets. In the cases where the model converged, best results of Stages over \model were observed in the YooChoose data (sub-sampled to 10k transitions) where Stage's MRR value is 0.15 while \model's MRR is 0.25. That is, in its best-performing dataset Stages is 40\% less accurate than \model. Moreover, Stages achieved MRR values of 0.05 for Brighkite (against 0.57 on \model) and of 0.10 on FourSQ (against 0.14 for \model).
%
%
%, using the first  70\% of transitions to perform training and the last 30\% for tests.  However, given that both FPMC and PRLME require parameter tuning, we leave-out 10\% of transitions of the training set as a validation set to optimize parameters. All methods were trained with $K=1024$ (the initial guess for \model), and baseline methods were trained using 1000 stochastic gradient descent iterations. The reason why we focus on smaller sets of data is, once again, the lack of scalability of the baseline methods. By
%

The results discussed so far consider only average measures of the effectiveness of the rankings produced by the methods. Going a step further, we also performed a Kolmogorov-Smirnov test between the distributions of reciprocal rank values obtained by \model and the competing methods.
The results again clearly indicate that \model obtains larger RR values over all datasets ($p < 0.001$), which is consistent the MRR results.
Another aspect is the impact of the walker sequence length $B$ on \model performance.
We test \model with $B = 1,2,3,4,5$ finding that larger $B$ improves accuracy at FourSQ, LFM-1k, LFM-Groups and reduces accuracy at Bkite and YooChoose as shown in Figure~\ref{f:B}.
As the best choice of $B$ is application dependent, we recommend fixing $B=1$.

\begin{figure}[ttt]
    \centering
    \includegraphics{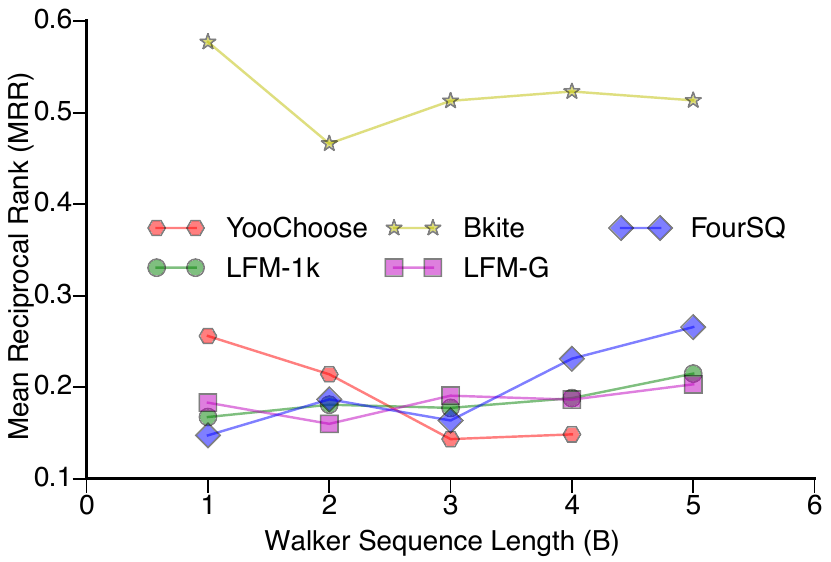}
    \vspace{-20pt}
    \caption{Impact of walker sequence length $B$.} 
    \vspace{-10pt}
    \label{f:B}
\end{figure}

\subsection{\model's Execution Time} \label{subsec:markov}
In this section we compare the performance of \model against Multi-LME~\cite{Chen2012,Chen2013} in our smallest datasets, namely Yes and FourSQ. {\bf These experiments validate the performance (in training time) of \model against another fully parallelized method.} 
For the sake of fairness, we also compare \model with MultiLME using the next-item likelihood measure: $P[x_{u,t+1} \mid x_{u,t}]$ as in Chen et al.~\cite{Chen2013} and other naive approaches (see below). %To compute this measure with \model,
The predictive likelihood using \model can be computed as:
%
%we define transition matrix between items, ${\bf P}^{(\text{TF})}$, as:
\begin{align} \label{eq:ourll}
    P[x_{u,t+1} \mid x_{u,t}] = \sum_{\cM} \mathbf{P}_{\cM}(x_{u,t+1}) \mathbf{P}_{\cM}(x_{u,t}) P[\cM]
    %{\bf P}^{(\text{TF})} = \sum_{\cM} {\bf P}_\cM \sum_{u} \pi_{\cM|\alpha\zeta,u} P[u] \, ,
\end{align}
\noindent where $P[\cM] \propto$ the number of tuples assigned to $\cM$ in the inference.
%where $P[u] = (\text{no.\ tuples of user $u$})/(\text{total no.\ tuples in dataset})$ and ${\bf P}^{(\text{TF})}({x_{t}, x_{t+1}})= P[x_{t+1} | x_{t}]$.
%$P[x_{t+1} \mid x_{t}]$ represents the probability of a given transition based on the random environments extracted by each method. 
Accuracy on the test set is measure using the log likelihood:
%  on our testing set $\cD_{test}$, we can look at the source $x_{t}$ of each transition in $\cD_{test}$ and estimate $P[x_{t+1} \mid x_{t}]$ for the destination $x_{t+1}$. More accurate models will achieve higher values of $P[x_{t+1} \mid x_{t}]$. If we sum up this result for each transition, we compute what is called the {\it predictive log likelihood:}
\begin{align} \label{eq:predll}
    \text{PredLL}(\cD_\text{future}) = \sum_{x_{u,t+1}, x_{u,t} \in \cD_\text{future}} log(P[x_{u,t+1} \mid x_{u,t}]).
\end{align}
%where $\cD_\text{future}$ is the test set defined by Chen et al.~\cite{Chen2013} in the Yes dataset and 30\% of the future transitions for other datasets.

\myparagraph{Naive Approaches}
To compute $P[x_{u,t+1} \mid x_{u,t}]$ one can also trivially adapt both  Latent Dirichlet Allocation (LDA)~\cite{Blei2003} and Transition Matrix LDA (TM-LDA)~\cite{Wang2012} for the same task.
In LDA each we define users as ``documents'', environments are ``topics'', and $x_{u,t}$ is the $t$-th word. 
TM-LDA follows the same definitions.
LDA is trained with $K$ latent factors and TM-LDA creates a $(K,K)$ matrix capturing the probability of transitioning between a latent ``topics''. 
Note that in these models $P[x_{u,t+1} \mid x_{u,t}] = P[x_{u,t+1}]$.
We refer to these adaptations as LDA' and TM-LDA' for simplicity as they do not reflect the original applications of LDA and TM-LDA.
%That is, if we define one of the dimensions of this matrix as the sources, $\cM_{x_t}$, and other dimension as the destination $\cM_{x_{t+1}}$, we can incorporate this transition matrix to capture transitions between items using TM-LDA as: $P[x_{t+1} \mid x_{t}]_{TMLDA} = \sum_{\cM_{x_{t+1}}} P[x_{t+1} \mid \cM_{x_{t+1}}] \sum_{\cM_{x_t}} P[\cM_{x_t} \mid x_t] P[\cM_{x_{t+1}} \mid \cM_{x_t}]$.

\myparagraph{Inference Procedure} 
In our evaluations we use the fully parallelized Multi-LME implementation~\cite{Chen2013}\footnote{\small\url{http://www.cs.cornell.edu/People/tj/playlists/}}. The LDA' and TM-LDA' adaptations to our problem are trained using \texttt{scikit-learn}~\cite{scikit-learn} which includes a fast/online~\cite{Hoffman2010} and parallelized implementation of these methods. While faster LDA' training methods do exist~\cite{Yu2015}, we preferred to make use of a mature software package. 
%More importantly, due to the similar algorithm, \model can incorporate some of the heuristics used to speed-up LDA' approaches~\cite{Yu2015}, an effort left as future work. 

The Yes dataset of Chen et al.~\cite{Chen2012,Chen2013} has no timestamps, thus we use \model-NT for this comparison instead of the more accurate general \model method. 
%We also evaluate both \model-NT and MultiLME on the FourSquare dataset since {\it it was the only other dataset where we were able to execute the MultiLME approach}. 
Each method was trained using different values of $K$, the number of \lowrank ($K \in \{10, 50, 25, 100\}$) on the Yes dataset and with $K=10$ on FourSQ because the original Multi-LME code has difficulties scaling to more factors on FourSQ. 

%\ju{Please, check all text for consistency with the verbs. Sometimes you use present tense. Others, you use past. Keep the past tense. It make is easier.}

\begin{figure}[t]
    \centering
    \includegraphics[scale=.95]{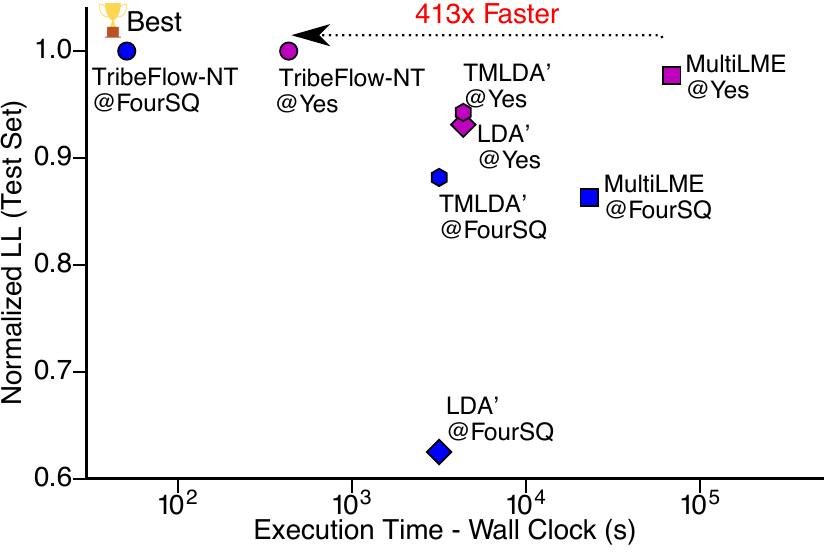}
    \vspace{-10pt}
    \caption{\model outperforms competition in next-item predictive log-likelihood task (small dataset due to competitors).}% \ju{(change axes:  x:  Training Time,  y: Normalized Likelihood)}}%Comparing \model with LME~\cite{Chen2012,Chen2013}, LDA~\cite{Griffths02} and TMLDA~\cite{Wang2012}.}
    \label{fig:lme_comp}
    \vspace{-10pt}
\end{figure}

\begin{figure*}[t]
    \centering
    \subfigure[``Classical Music'']{\frame{\includegraphics[scale=.24]{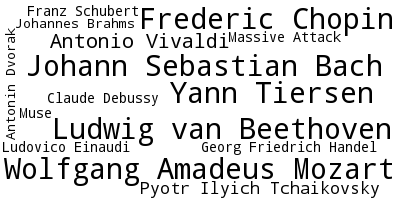}}}
    \subfigure[``Movie Composers'']{\frame{\includegraphics[scale=.24]{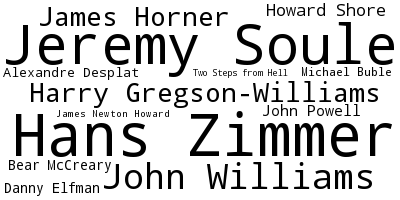}}}
    \subfigure[``Rock'']{\frame{\includegraphics[scale=.24]{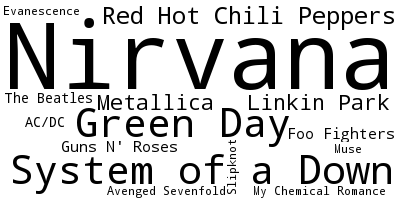}}}
    \subfigure[``Heavy Metal'']{\frame{\includegraphics[scale=.24]{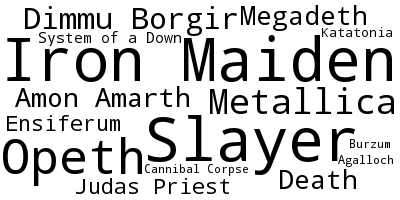}}}
    \subfigure[``Electro House'']{\frame{\includegraphics[scale=.24]{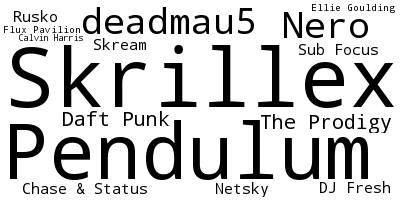}}}
    \vspace{-10pt}
    \caption{Examples of popular items in latent \model  environments in a dataset dominated by pop music fans (LFM-G).}
    \label{fig:lastfmtrans}
\end{figure*}

\myparagraph{Results} 
Figure~\ref{fig:lme_comp} presents our results in the next-item predictive log-likelihood task using $K=100$ and  $K=10$ on the Yes and FourSQ datasets. 
We get similar results for Yes with $K \in \{10, 25, 50\}$. 
For the sake of interpretability we normalize the predictive log likelihood of each method by the best result (always \model).
Each point in the figure represents one method-dataset pair, while different datasets are represented by different colors. We labelled each point in order to help the interpretation of the results. With these settings, it is expected that the method that performs the best embedding is placed in the top-left corner of the figure (higher accuracy, lower runtime). 
The x-axis  represents the execution time, whereas the y-axis represents the normalized likelihood. 
%Also, on this figure we only present results for $K=100$ on Yes and $K=10$ on FourSquare. 
%on th Yes datasets and MultiLME can only handle this setting for FourSquare. 

In Figure~\ref{fig:lme_comp} we clearly see that \model-NT is the best approach on both datasets. 
Compared to MultiLME \model-NT achieves a higher log-likelihood at a fraction of the runtime (speedups are up to 413$\times$). 
As expected, \model-NT and MultiLME usually outperforms the LDA-based baselines since \model-NT and MultiLME were built to explicitly capture user trajectories. 
Interestingly, comparing the speed of \model and  \model-NT in Figures~\ref{fig:cjewel} and~\ref{fig:lme_comp} we see that for FourSQ \model-NT is one order of magnitude faster than \model but less accurate than \model. This occurs because \model-NT does not infer inter-event time distributions using the ECCDF heuristic.
\subsection{Flows of Users Between Locations} \label{subsec:gravity}
\model outperforms the state-of-the-art methods in sophisticated tasks such as ranking and predictive next-item likelihoods. 
But what about a simpler task? Can \model outperform simple application-specific methods?
In this section we compare \model against the Gravity Model (GM) for uncovering the average flows of users between two locations.
The widely popular Gravity Model (GM) uses GPS coordinates and requires a pre-defined distance function, $dist$, capturing the proximity of locations around the globe. 
As in Smith et al.~\cite{Smith2013} and Garcia-Gavilanes et al.~\cite{Garcia-Gavilanes2014} we employ the distance on a sphere from the latitude and longitude coordinates of the venues in our LBSN datasets.
The three parameters, $\theta_1$, $\theta_2$ and $\theta_3$ of the distance are fitted using a Poisson regression that is known to lead to better results~\cite{Silva2006}.

Our goal is to estimate $f_{ds}$, the flow of users going from location $s \in \Omega$ to location $d \in \Omega$.
Let $n_d = \sum_{x_{u,t+1}, x_{u,t} \in \cD_{present}} {\bf 1}(x_{u,t+1} = d)$  denote the number of visits of all users to a destination $d$ and $r_s$ be the equivalent number of visits of all users to a source location $s$.
GM captures the flows of users between the two locations $f_{ds}$ as
$$\hat{f}^{(\text{GM})}_{ds} = \frac{r_s^{\theta_1} n_d^{\theta_2}}{dist(d, s)^{\theta_3}}.$$

\model can trivially estimate the flow $\hat{f}^{(\text{TF})}_{ds}$ from Eq.~\eqref{eq:ourll}.
Since gravity models are limited to geolocated datasets (need GPS coordinates), we compare \model with GM on  our FourSQ and Brightkite. 
Note that unlike GM, \model is application-agnostic and {\bf does not} use GPS coordinates, albeit it would be straightforward to incorporate such application-specific features in the latent environments.
Further, we opt to use \model-NT (instead of the more accurate \model method) because of its slightly faster running time, trading-off accuracy for speed. 
We infer the posteriors from \model-NT with $K=10$ environments. 
In this setting, training both models takes less than 5 minutes.  
%This implies that, even better results could be achieved if the more \ju{complete} solution was adopted. \ju{ More complete??? How do you compare \model and \model-NT? please rephrase the adjective. Also last sentence is incomplete. you could obtain better results in terms of acccuracy but what about training time? Please change sentence to consider both metrics.} We trained \model-NT with  $K=10$. In this setting, the training of both \model-NT and GM takes less than 5 minutes. 
Methods are evaluated using the mean absolute error (MAE).%$$\text{MAE} = \frac{\sum_{d,s}|f_{ds} - \hat{f}^{(\cdot)}_{ds}|}{|\{(d,s) | \forall x_{t+1},x_t \in \cD_\text{future}\, (d = x_{t+1} \land s = x_t) \}|}.$$

\model-NT significantly outperforms GM for geolocation flows with just a few ($K=10$) random environments and, unlike GM, without GPS coordinates. 
Specifically, GM achieves MAE results of 10.48 and 9.81 on the Bkite and FourSQ datasets, respectively. 
In contrast,  \model-NT achieves {\bf 1.606} and {\bf 1.41} on Bkite and FourSQ, respectively.  
{\it The improvements of \model range from roughly 800\% to 900\%} in mean absolute error. 
Again, validating our results using the Kolmogorv-Smirnov test showed that \model-NT is statistically more accurate than GM ($p < 0.001$). 

%Having seen that \model is faster and more accurate for next-item ranking and prediction than state-of-the-art methods, in the next section we discuss how \model can also be used for exploratory data analysis.

\section{Exploratory Analysis}
\label{sec:exploratory}
%!TEX root = paper.tex
%In the previous section compared \model against several state-of-the-art approaches. 

\begin{figure}[t]
    \centering
    \includegraphics[scale=.42]{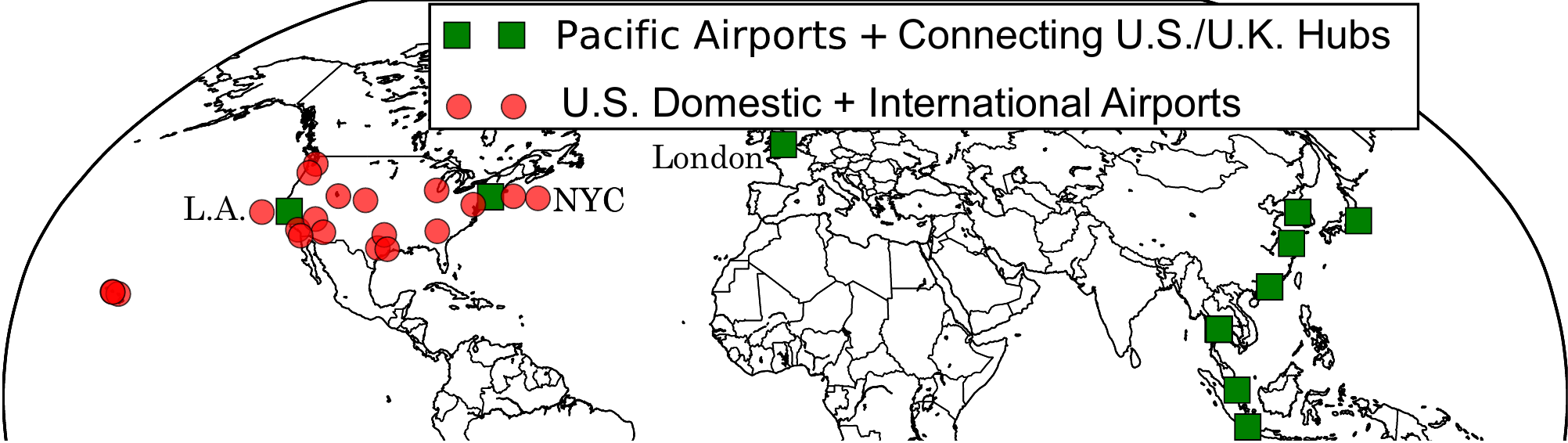} \\
    \caption{Latent flow environments inferred by \model without GPS information.}
    \label{fig:geo}
\end{figure}

\begin{figure}[t]
\centering
\subfigure[User 1]{\includegraphics[scale=.75]{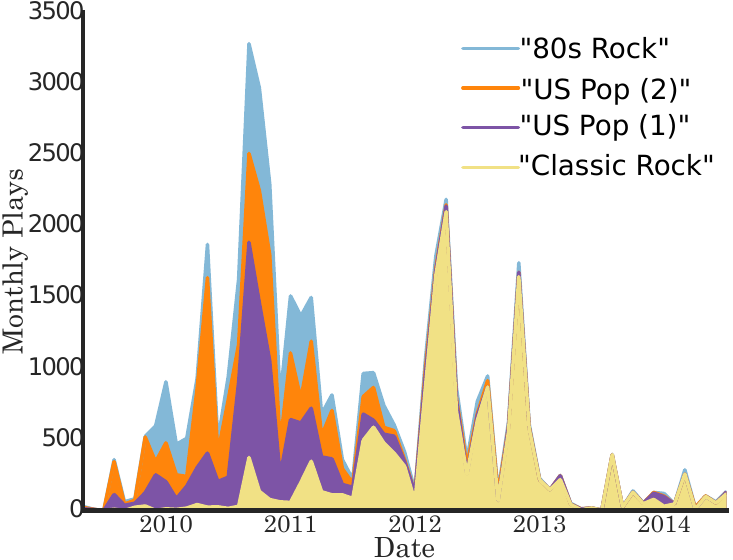}}
\subfigure[User 2]{\includegraphics[scale=.75]{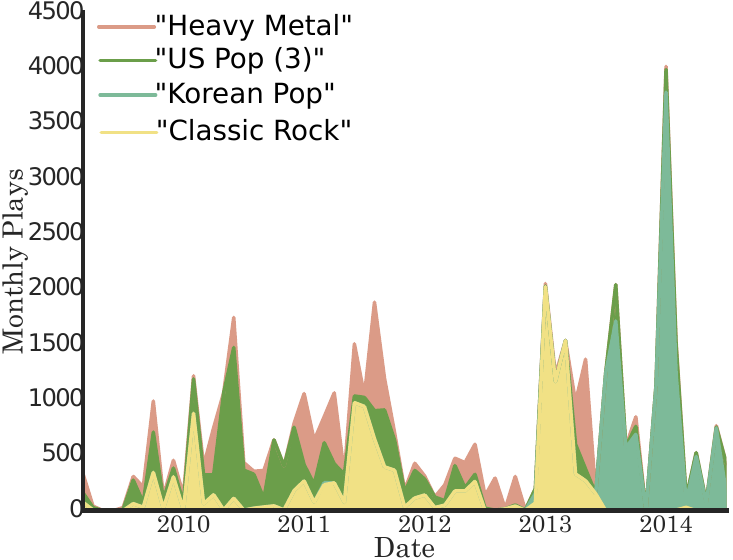}}
\vspace{-10pt}
\caption{Example of Two Users From Last.FM-Groups}\label{f:ua}
\vspace{-10pt}
\end{figure}

In this section we consider how \model can be used for sense-making in our datasets. 
Section~\ref{subsec:artists} discusses the semantics of \lowrank inferred by \model in our LastFM-Groups dataset, specially in the presence of non-stationary, transient, and time-heterogeneous user trajectories. 
Section~\ref{subsec:geosemantics} introduces latent environments inferred in the FourSQ dataset without GPS data. 
Finally, Section~\ref{subsec:temptensors} discusses how \model compares with tensor decomposition approaches when uncovering meaningful patterns of user behavior from stationary user trajectory data.

\subsection{Artist-to-Artist Transitions} \label{subsec:artists}
Figure~\ref{fig:lastfmtrans} shows five latent environments inferred by \model from the LastFM-Groups dataset. 
Each latent environment is represented as a word-cloud of the names of the top 15 artists in the environment ranked by the random walk steady-state probability at each environment. 
We cross-reference the top artists at each latent environment in Figures~\ref{fig:lastfmtrans}(a-e) with the AllMusic guide\footnote{\url{http://www.allmusic.com/}}, finding that the environments discovered are semantically meaningful. 
It is worth noting that users in the LastFM-Groups dataset are overwhelmingly declared fans of pop music (as described in Section~\ref{subsec:data}), but even then, \model is able to extract the user interests of very diverse musical themes such as those depicted in Figure~\ref{fig:lastfmtrans}.

Taking a more in-depth look at the latent environment represented in Figure~\ref{fig:lastfmtrans}(a) shows a sequential user trajectory preference for songs related to ``Classical Music''. 
Some of the main artists/composers in this environment are: Chopin, Bach, Beethoven and Mozart. 
The latent environment in Figure~\ref{fig:lastfmtrans}(b) shows composers of motion picture sound tracks. 
John Williams, for instance, compose soundtracks for popular movies such as Star Wars, Jaws, ET and Superman. 
The environment in Figure~\ref{fig:lastfmtrans}(c) represents popular rock bands such as Nirvana and Green Day, whereas the environment in Figure~\ref{fig:lastfmtrans}(d) represents heavy metal bands such as Iron Maiden and Slayer. 
Finally, the environment in Figure~\ref{fig:lastfmtrans}(e) represents {\em electro house} music being composed of groups such as Skrillex, Pendulum, deadmau5, and Nero.

%The main artists is, as expected ``The Beates'', but various related artists also appear in the environment. For instance, ``Paul Mccartney'' and ``Paul Mccartney \& The Wings''represent musical tracks from one of The Beatles leading members and former band. ``Evan Rachel Wood'' is an actress and ``Jim Sturgess'' is an actor both cast in the movie {\it Across the Universe (2007)}\footnote{\url{http://www.imdb.com/title/tt0445922/}} in which both performed various songs from ``The Beates''. The presence of artists such as ``Lady Gaga'' and ``Britney Spears'' to a smaller extent reveals the bias on this datasets towards such artists (based on the Last.FM groups sampled). However, the very fact that \model is able to extract a group of artists which are related to one another not a as genre (e.g., pop or rock), but as a {\it theme} -- songs related to or by The Beatles -- represents the strength of \model when applied to exploratory data analysis. 

%Regarding the other environments, Figure~\ref{fig:lastfmtrans}(b) depicts composers from motion picture sound tracks.  Figure~\ref{fig:lastfmtrans}(c) represents popular rock artists, whereas Figure~\ref{fig:lastfmtrans}(d) represents heavy metal bands. Finally, Figure~\ref{fig:lastfmtrans}(d) shows a user environment focused on electronic music.  {\em In a dataset mostly comprised of pop artists fans (Last.FM-Groups), \model is able to extract the user interests of diverse musical themes as depicted in Figure~\ref{fig:lastfmtrans}}.

\subsection{Flow Semantics in Check-in Data} \label{subsec:geosemantics}

We now turn our attention to the Foursquare dataset. Environments in this dataset capture user trajectories between businesses. 
We want see whether \model infers semantically meaningful latent environments for check-in trajectory data.
Interestingly, because \model~{\bf does not use GPS} features, \model can identify latent environments of ``nearby'' locations in any geometry.

Figure~\ref{fig:geo} shows the top-20 locations for two different latent environments inferred by \model. 
Each point in the figure is a latitude-longitude coordinate plotted on the world map from its GPS coordinate. 
%To gain further insight into these coordinates, we queried the coordinates in Google Maps\footnote{\url{http://maps.google.com}} and determined the location of each point.  
Two latent environments best exemplify the sense-making abilities of \model. 
The first latent environment is represented by U.S.\ airports and seems to capture flows of U.S.\ domestic flights (including Hawaii). 
The second latent environment captures check-in trajectories of connections to/from Pacific-Asia-based airports. Also in this environment are major U.S./U.K.\ major hubs (JFK in New York City (NYC), LAX in Los Angeles and Heathrow in London) that connect to Pacific-Asia. 
%between U.S./U.K.\ major hubs (JFK in New York City (NYC), SFO in San Francisco and Heathrow in London) to/from Pacific-Asia-based airports. 
Although omitted from the figure, \model also extracted environments based on user trajectories within cities including NYC, Miami and Atlanta. In these settings, check-ins are related to different places in these cities. 
This illustrates that \model latent environments can be a powerful tool for sense-making in user trajectory datasets.

\subsection{Transient User Trajectories} \label{subsec:temptensors}
In this section we illustrate how \model can help us identify transient user trajectories (also showing that \model can cope well with the transience). 
To illustrate this, Figures~\ref{f:ua}(a,b) present the number of song plays (y-axis) over multiple years (x-axis) of two users from Last.FM-Groups broken down into the user's four preferred latent environments. 
More precisely, the y-axis shows the cumulative sum of the song plays at a given month color-coded by the song's most likely latent environment for that user, $P[\cM | u, x_{t}]$. 

As shown in the Figure~\ref{f:ua}(a),  from 2010 until mid 2011 the user goes through a strong Pop phase -- most representative (top)  artists in the environment labeled ``U.S.\ Pop (1)'' are  {\it Madonna, Nelly Furtado,} and {\it Alicia Keys},  and top artists in the  ``U.S.\ Pop (2)'' environment are  {\it Britney Spears, Leona Lewis, } and {\it Kelly Clarkson}. 
We also note some interest in 70-80's Rock overtones -- top artists being  {\it Queen, Michael Jackson, } and {\it The Beatles}.
After mid 2011 the user moves away from Pop artists towards a ``Classic Rock'' environment, with  {\it The Beatles, Pink Floyd,} and {\it  Nirvana} as top artists. 

The user represented in Figure~\ref{f:ua}(b) also changes interest over time, most markably from ``Classic Rock'' \& ``Heavy Metal'' to ``Korean Pop''. 
The three major take-aways are: (a) user trajectories are indeed transient (Figure~\ref{f:ua}); (b) users can show interests in the same latent environment at different points in time: User 1 Figure~\ref{f:ua}(a) shows strong ``Classic Rock'' environment preference between 2011-2012 while User 2 Figure~\ref{f:ua}(b) shows strong preference for the same environment between 2013-2014; and (c) \model can cope well with transient trajectories.

%Through the illustrative example in Figure~\ref{f:ua} we see that the trajectories of a given user is evolving over time. %This figure represents the most popular environment of artists played by a specific user over time.  
%To exemplify this, we looked into the registration dates of different users in our Last.FM-1k dataset. We found that, on average, a new user will arrive on the application every other day (2 days between arrivals), with 50\% of the user registrations happening in a day or even less. In other words, new users arrive at the application at different times. Thus,  it is highly unlikely that users with the same tastes will have synchronous behavior. In fact, we computed the difference between the registration dates of the top 50 users in each latent factor, obtaining an average of  $339.4$ days (standard deviation of $264.9$ days). In this sense, even the top users interested in a given musical genre, captured by a latent factor, will have registered in Last.FM roughly 6 months apart from each other, on average. This provides some evidence that their behavior, though similar, might not be synchronous in time.  In such case, a temporal tensor decomposition would not capture such similarities, as discussed below. 

To provide further evidence in support \model's ability to extract stationary, ergodic, and time-homogeneous behavior from user trajectories consider the following synthetic dataset.
For comparison, we contrast \model with a state-of-the-art temporal tensor approach in the same scenario.
This comparison sheds light into the reasons why tensor-factorization-based methods should have difficulty in extracting stationary, ergodic, and time-homogeneous behavior from user trajectory data.

Our synthetic dataset has 50 users event assigned to one of five Markov chains.
The Markov chains are used to create user trajectories. 
We model the popularity of items at each chain as a Lognormal distribution.
Also, each environment has exponentially distributed inter-event times but timestamps are not recorded. 
We simulate a total of 5 days where and each user selects in average 100 items per day.  
We simulate users joining the system 1 or 2 days apart.

%
%Starting from a random item, for each transition, we sampled a destination item from the preference distribution of the user group with probability 0.99. We sampled the destination from a different distribution with probability 0.01 (to introduce some noise). 

%Most approaches that mine user behavior over time represent the data as tensors with an explicit time mode. 
%Such methods do well in capturing synchronous behavior~\cite{Matsubara2012a}. 
%However, such approaches are not meant for sense-making with transient trajectories.
%
%
%, we show that, unlike temporal tensor approaches, \model is able to more accurately recover the inherent preferences of each user. 
 
\begin{figure}[t]
    \centering
    \includegraphics[scale=.9]{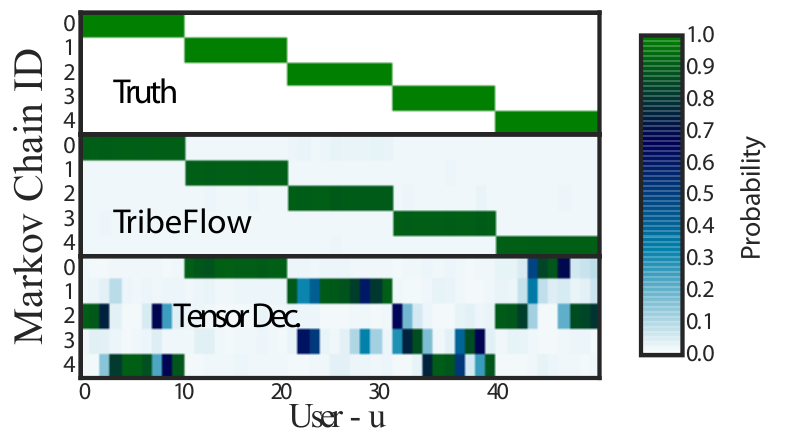}
    \vspace{-10pt}
    \caption{Comparison with Temporal Decomposition.}
    \vspace{-5pt}
    \label{fig:trimine}
\end{figure}

Applying \model-NT to this synthetic data almost perfectly recovers user preferences of transition matrices as presented in Figure~\ref{fig:trimine}. \model should work even better if timestamps were available, but that would be an unfair comparison because of the extra feature.
A state-of-the-art temporal decomposition method~\cite{Matsubara2012a} (tensor with time mode) on the same data has trouble finding the ground truth (i.e., user ids  $0-9$ using Markov chain 0, user ids $10-19$ using Markov chain 1, etc). 
Note that only \model is able to recover the different user trajectory preferences despite the asynchronous behavior of users.
Clearly this type of tensor decomposition is meant to uncover only synchronized behavior as originally proposed by Matsubara et al.~\cite{Matsubara2012a}.

\section{Conclusions}
\label{sec:conclusions}
%!TEX root = paper.tex
In this work we introduced \model, a general method to mine and predict user trajectories. 
\model decomposes non-stationary, transient,  time-heterogeneous user trajectories into a small number of short random walks on latent random environments.
The decomposed trajectories are stationary, erratic, and time-homogeneous in short time scales. 
User activity (e.g., listening to music, shopping for products online or checking-in different places in a city) is then captured by different latent environments inferred solely from observed user navigation patterns (e.g., listening to ``classical music'', listening to ``Brazilian Pop'', shopping for shoes, shopping for electronics, checking-in into airport fast-food venues, etc.).
We summarize our major contributions as follows:
\begin{itemize}[itemsep=-0.1mm,leftmargin=*]
    \item {\bf Accurate:} \model outperforms various state-of-the-art baseline methods in three different tasks: extracting Markov embedding of user behavior, next-item prediction and capturing the flows of users between locations. Gains are up-to 900\% depending on the dataset and task analyzed. %\model is more accurate than Multi-LME by Chen et al.~\cite{Chen2013} when capturing the likelihood of trajectories on the very datasets used in Chen et al.~\cite{Chen2013} \model outperforms FPMC~\cite{Rendle2010a} and PRLME~\cite{Feng2015} in mean reciprocal rank (at least 23\% gain), and beats the Gravity Model~\cite{Silva2006} when measuring the amount of times users transitions between geographic locations (up to 900\% gain).  
\item {\bf Scalable: } \model is at least tens and up to hundreds of times faster than state-of-the-art competitors even in relatively small datasets. If we consider the only other fully parallel competitor, MultiLME~\cite{Chen2013}, \model is 413x faster and still more accurate.
%\item {\bf Parameter-free:} \model has no tunable parameters unlike its state-of-the-art competitors. %For instance, with \model we were able to extract ``niche user preferences'' (e.g., songs from movie composers) from a dataset highly biased towards pop music.
\item {\bf Novelty:} \model provides a general framework to build upon (random surfer over infinite latent random environments) and make application-specific personalized recommendation systems.
\end{itemize}
%\vspace{-5pt}
%As future work we aim at extending the model to exploit even faster heuristics~\cite{Yu2015} and different, possibly more accurate, learning methods (e.g., such as a reversible jump MCMC)~\cite{Gelman2013}. Also, based on some initial experiments we found that \model can help different application such as hypothesis testing~\cite{Singer2015}, where we can use the latent environments to test {\it multiple hypothesis} about different, but overlapping, sub-populations of a dataset. 

\section*{Acknowledgments} 
Research was funded by Brazil's National Institute of Science and Technology for Web Research (MCT/CNPq/INCT Web 573871/2008-6). Research was also sponsored by the Defense Threat Reduction Agency and was accomplished under contract No. HDTRA1-10-1-0120, as well as by the Army Research Laboratory and was accomplished under Cooperative Agreement Number W911NF-09-2-0053. The views and conclusions contained in this document are those of the authors and should not be interpreted as representing the official policies, either expressed or implied, of the Army Research Laboratory or the U.S. Government. The U.S. Government is authorized to reproduce and distribute reprints for Government purposes notwithstanding any copyright notation here on.

\balance
\bibliographystyle{abbrv}
\bibliography{BIB/myref.bib}

\end{document}